\documentclass[12pt, draftclsnofoot, onecolumn]{IEEEtran}%

\ifCLASSINFOpdf
  \usepackage[pdftex]{graphicx}
  \graphicspath{{../pdf/}{../jpeg/}}
  \DeclareGraphicsExtensions{.pdf,.jpeg,.png}
\else
  \usepackage[dvips]{graphicx}
  \graphicspath{{../eps/}}
  \DeclareGraphicsExtensions{.eps}
\fi

\usepackage{epsfig}
\usepackage{booktabs}
\usepackage{epstopdf}
\usepackage{mathtools}

\usepackage{array}

\usepackage{bm}
\usepackage{mathrsfs}
\usepackage{epsfig}
\usepackage{subfigure}
\usepackage{algorithmic}
\usepackage{algorithm}
\usepackage{mathrsfs}

\usepackage{amssymb}
\usepackage{amsmath}
\usepackage{cite}
\usepackage{url}
\usepackage{xcolor}
\usepackage{cite,graphicx,amsmath,amssymb}
\usepackage{subfigure}
\usepackage{fancyhdr}
\usepackage{mdwmath}
\usepackage{mdwtab}
\usepackage{caption}
\usepackage{amsthm}
\newcolumntype{P}[1]{>{\centering\arraybackslash}p{#1}}

\newtheorem{remark}{Remark}
\newtheorem{theorem}{Theorem}

\newtheorem{lemma}{Lemma}

\newtheorem{corollary}{Corollary}

\newtheorem{proposition}{Proposition}

\begin{document}
\title{Joint Reflecting and Precoding Designs for SER Minimization in Reconfigurable Intelligent Surfaces Assisted MIMO Systems}

\author{
Jia~Ye,~\IEEEmembership{Student Member,~IEEE}, Shuaishuai Guo,~\IEEEmembership{Member,~IEEE},  \\and Mohamed-Slim Alouini, ~\IEEEmembership{Fellow,~IEEE}
\thanks{J. Ye and M.-S. Alouini are with Electrical Engineering, Computer Electrical and Mathematical Sciences $\&$ Engineering (CEMSE) Division, King Abdullah University of Science and Technology (KAUST), Thuwal, Makkah Province, Kingdom of Saudi Arabia, 23955-6900 (email:jia.ye@kaust.edu.sa; slim.alouini@kaust.edu.sa).

S. Guo was with CEMSE Division, KAUST, Thuwal, Makkah Province, Kingdom of Saudi Arabia, 23955-6900. He is now with Shandong Provincial Key Laboratory of Wireless Communication Technologies and and School of Control Science and Engineering, Shandong University, Jinan 250061, China (email: shuaiguosdu@gamil.com).}

}

\maketitle

\begin{abstract}
This paper investigates the use of a reconfigurable intelligent surface (RIS) to aid point-to-point multi-data-stream multiple-input multiple-output (MIMO) wireless communications. With practical finite alphabet input, the reflecting elements at the RIS and the precoder at the transmitter are alternatively optimized to minimize the symbol error rate (MSER). In the reflecting optimization with a fixed precoder, two reflecting design methods are developed, referred as eMSER-Reflecting and vMSER-Reflecting. In the optimization of the precoding matrix with a fixed reflecting pattern, the matrix optimization is transformed to be a vector optimization problem and two methods are proposed to solve it, which are referred as MSER-Precoding and  MMED-Precoding. The superiority of the proposed designs is investigated by simulations. Simulation results demonstrate that the proposed reflecting and precoding designs can offer a lower SER than existing designs with the assumption of complex Gaussian input. Moreover, we compare RIS with a full-duplex Amplify-and-Forward (AF) relay system in terms of SER to show the advantage of RIS.
\end{abstract}



\section{Introduction}
Better communication qualities like smaller delay, higher transmission rate, lower symbol error probability (SER), less energy consumption always attract researchers' and users' eyes. In the last decade, wireless networks have been greatly improved thanks to various technological advances, including massive multiple-input multiple-output (Massive MIMO), millimeter wave (mmWave) communications, and ultra-dense deployments of small cells. However, some critical issues such as the hardware complexity and system update cost are still blocking their steps to the practical implementation \cite{SZ}. To satisfy the growing demands with satisfying communication quality and achieve challenge goals in a green and effective way, reconfigurable intelligent surface (RIS) was proposed as a promising solution in the coming 5G or beyond era \cite{MDR}. 

RIS is a planar array comprising of a large number of nearly passive, low-cost, reflecting elements such as positive-intrinsic-negative (PIN) diodes, which are used for altering the phase of the reflected electromagnetic wave with reconfigurable parameters and smart controller. RIS can be implemented by various materials, including reflect arrays \cite{XT,SV}, liquid crystal metasurfaces \cite{SF}, ferroelectric films, or even metasurfaces\cite{CL}.   

In the beginning, reflecting surfaces were not considered in wireless communication systems because these surfaces only had fixed phase shifters, which could not adapt the phase modification in time-varying wireless propagation environments. Recently,  advanced micro-electrical-mechanical systems (MEMS) and metamaterials have been investigated as a solution to this issue, which enables the real-time reconfiguration reflecting surfaces \cite{TJ}. Compared to existing related technologies such as multi-antenna relay \cite{BS}, backscatter communication\cite{GY} and active intelligent surface-based massive MIMO \cite{SH}, passive RIS does not require any dedicated energy source for either decoding, channel estimation, or transmission. It only reflects the ambient radio frequency (RF) signals in a passive way without a transmitter module. Moreover, the reflect-path signal through RIS carries the same useful information as well as the direct-path signal without any information of its own, which will not cause any additional interference. 

RIS stands out among these technologies by smartly adjusting the phase shifts induced by all the elements with advantages like overcoming unfavorable propagation conditions, enriching the channel with more multi-paths, increasing the coverage area, improving the received signal power, avoiding interference, enhancing security/privacy and consuming very low energy. On the other hand, the lightweight and conformal geometry of RIS can enable the installment onto the facades of buildings in outdoor communication environments or the ceilings and walls of rooms in indoor communication environments, which provides high flexibility and superior compatibility for practical implementation \cite{LS}. Also, integrating RIS into the existing networks can be made transparent to the users without the need for any change in the hardware and software of their devices.

\subsection{Prior Work}
Due to the advantages mentioned above, RIS has attracted more and more researchers' attention in the last few years. Many interesting and contributed works appeared in the industry \cite{NTT} and academia field regarding to channel estimation \cite{BZ}, \cite{ZQ}, multi-cell networks \cite{CP}, comparison with relaying systems \cite{EBR} and massive MIMO networks \cite{EBM}, combination with other technologies like SWIPT \cite{QW3}, \cite{CP1}, millimeter wave \cite{XT}, \cite{PW}, Terahertz communication \cite{WC} and so on. A large and growing body of literature has optimized the RIS parameters and systems' structure to improve system performance, such as outage probability \cite{XT}, signal power \cite{QW,QW2,QI,CH,CH1}, signal-to-interference-plus-noise ratio (SINR) \cite{QI}, data rate \cite{CP}, \cite{YH,MJ,MJ1,MJ3,XT1,HG,SA}, spectral efficiency (SE) \cite{CH,CH1,XT1}, secrecy achievable rate \cite{CH2,GX,HS,MC} and so on. To describe the status of research, we survey some of them as follows before introducing our work.

In \cite{XT}, the authors enhanced the mmWave link robustness and optimized the link outage probability by deploying smart reflect-arrays when the direct links are blocked by obstructions. The authors also investigated the optimal beam direction for randomly moving devices without any location information, incorporating the antenna sector selection at the access point (AP) and the mobile user as well as the configuration of the smart reflect-arrays. 

For maximizing the total received signal power, preliminary contributions appeared in \cite{QW,QW2}. Specifically, a centralized algorithm with the global channel state information (CSI) and a low-complexity distributed algorithm for designing the phase shifts were proposed in \cite{QW}. With the same system setup in \cite{QW}, the authors of \cite{QW2} considered a more practical case when the RIS only has a finite number of discrete phase shifts in contrast to the continuous phase shifts. It can be seen that the asymptotic squared power gain of RIS-aided multiple-input single-output (MISO) shown in \cite{QW} with continuous phase shifts still holds with discrete phase shifts with a stable performance loss gap between the two conditions. The authors of \cite{CH} and \cite{CH1} maximized the energy-efficiency (EE) and SE of a RIS-assisted multi-user MISO system by designing both the transmit power allocation at the BS and the phase elements of the RIS. 

By using the knowledge of only the channel large-scale statistics instead of the global CSI, \cite{QI} designed an optimal linear precoder and the power allocation at the base station (BS) as well as the RIS phase matrix to maximize the minimum SINR in multi-user RIS-assisted MISO communication systems. It was shown that the RIS-assisted system can achieve power gains with a much fewer number of active antennas at the BS. 

The authors in \cite{MJ,MJ1,MJ3} analyzed the approximated uplink ergodic rate of a Rician fading system and derived an optimal size of a RIS unit, where users are mapped to a limited area of the entire RIS. It was shown that RIS can bring improved reliability with a significantly reduced area for antenna deployment compared to massive MIMO. Targeting for practical implementation, \cite{EB} investigated an optimal phase shift design to maintain an acceptable degradation of the ergodic capacity. The authors in \cite{CH2} proposed a provably convergent, low-complexity method to maximize the system sum-rate. It was shown that a nearly “interference-free” zone can be established in the proximity of the RIS thanks to its spatial interference nulling/cancellation capability. The authors in \cite{CP} maximized the weighted sum rate in the RIS-assisted multicell MIMO system, where a RIS located at the cell boundary of multiple cells to assist the transmission and alleviate the inter-cell interference. 

Moreover, using RIS to secure wireless communications is also a hot topic in this field. The authors of \cite{GX} jointly optimized the transmit beamforming with artificial noise jamming to maximize the achievable secrecy rate. The physical-layer security was also studied in \cite{HS}, which aims to maximizing system secrecy rate under source transmit power constraint on the transmitter and the phase shifts unit modulus constraints on RIS. By jointly designing the transmit beamforming of the AP and the reflect beamforming of the RIS, the authors of \cite{MC} optimized the secrecy rate of the legitimate communication link considering the presence of an eavesdropper.

In all the studies reviewed here, RIS is a promising technology to satisfy the growing demand for data rates and communication quality. Current research on the joint optimization of beamforming and reflecting in RIS-aided communications paid more attention to maximizing SE, EE, signal-to-noise ratio (SNR) and so on assuming Gaussian-distributed signal input. However, the assumption with Gaussian-distributed input signals is impractical, since the Gaussian-distributed signals are not bounded. In practical, the widely adopted is finite discrete constellation signals, such as phase-shift keying (PSK) or quadrature amplitude modulation (QAM) signals. The authors in \cite{WY} showed that the probability distribution function of a standard Gaussian-distributed signal is significantly different from the probability mass function of a quadrature phase-shift keying (QPSK) signal.   Based on previous research on MIMO communications without the aid of RIS, the optimization with the Gaussian input assumption is far from optimal in MIMO communications with finite alphabet input, and also the optimization with finite alphabet input is much more complicated \cite{WY}. In our paper, we are interested in the joint optimization of beamforming and reflecting to minimize SER for RIS-aided MIMO communications assuming finite alphabet input. Although the authors in\cite{EB} and \cite{EB1} evaluated the error performance of RIS-based communication systems, they do not do the optimization regarding the phase shifts and precoder. To the best of our knowledge, there is no work to jointly optimize the reflecting and precoding for RIS-assisted MIMO communications based on the SER minimization criterion with finite alphabet input. The use of SER as the metric for performance optimization calls for a thorough investigation, which motivates our work.

\subsection{Contributions}
As was pointed out in the previous subsection, we are interested in the joint optimization of precoding and reflecting to minimize SER for RIS-aided MIMO communications assuming finite alphabet input. The optimization objective is different from existing works. The challenge of the optimization lies in that the beamforming and reflecting affect all received signal vectors. The optimization complexity increases with the square growth of the number of transmit signal vector candidates, because the mutual Euclidean distances among different noise-free received signal vectors jointly affect the SER. We transform the SER minimization problem into several problems that can be addressed by existing optimization techniques. 
For clearness, we list our contributions as follows:
\begin{itemize}
\item This paper considers a RIS-enhanced point-to-point multiple-data-stream MIMO communication system, where a multiple-antenna transmitter serves a multiple-antenna receiver with the help of a RIS. 
We formulate the SER minimization problem regarding the precoding and reflecting for such a system. Since the precoding and reflecting are coupled in affecting the objective function and isolated physically, we resort to an alternating strategy, which is widely used in literature (e.g., \cite{CP1, CP,HG,QWW}) to maximize SE, EE, SNR, etc.

\item Regarding the reflecting design, we reformulate the reflecting optimization to find the optimal reflecting angle subsequently to minimize the symbol error rate (MSER) and solve it by a coordinate descent method, referred as element-wise MSER (eMSER) reflecting. The eMSER-Reflecting is analyzed to have high computational complexity. To reduce the complexity, we transform a non-convex constraint of the problem to be a convex constraint and solve the optimal reflecting vector jointly through a vector gradient descent method, referred as vMSER reflecting. 

\item Regarding the precoding design, we express the channel matrix, the precoding matrix, and the transmitted symbol vectors in new forms and transform the precoding matrix optimization to be a vector optimization problem. The reformulated problem can be solved by a projected gradient descent method named MSER-Precoding. In the high SNR regime, the minimizing SER criterion is shown to be equivalent to the maximizing the minimum Euclidean distance (MMED) criterion. Based on the fact, the reformulated problem is reduced to be a quadratically constrained quadratic program (QCQP) problem and we propose a corresponding MMED-Precoding method. MMED-Precoding is shown to achieve comparable performance with considerable computational complexity reduction. 

\item Moreover, a thorough investigation is conducted by simulations. The effectiveness of our proposed algorithms is validated by comparing them with the solution gained through an exhaustive search in performance, computation complexity, and central processing unit (CPU) running time. By comparing the proposed designs with existing designs that maximize SNR, we verify that the SER minimization with finite alphabet input could not be replaced by SNR maximization, especially in the high SNR regime. Although the proposed designs are with higher computational complexity than existing designs, they can achieve much better performance in practical communication systems with finite alphabet input. Moreover, numerical comparison with a full-duplex Amplify-and-Forward  (AF) relay system is conducted to show that the RIS-assisted system can achieve the same or even better performance by adopting reflecting elements rather than additional power amplifiers at relays.
\end{itemize}

\subsection{Organization}
The remainder of the paper is organized as follows. Section II describes the system model and formulates the system optimization problem. Section III introduces the proposed reflecting and precoding designs. In Section IV, we analyze the computational complexity. Numerical comparisons are presented in Section V and conclusions are drawn in Section VI.

\subsection{Notations}
In this paper, $x$ denotes a scalar; $\textbf{x}$ represents a vector; $\textbf{X}$ stands for a matrix. ${\left\| \textbf{x} \right\|_2 }$, ${\left\| \textbf{x} \right\|_p }$ and ${\left\| \textbf{x} \right\|_\infty }$ represents $l_2$ norm, $l_p$ norm and $l_\infty$ norm of $\textbf{x}$ respectively. ${\left\| \textbf{X} \right\|_F }$ is the Frobenius norm of $\textbf{X}$.  $\text{diag}\left(\textbf{x}\right)$ is a diagonal matrix whose diagonal entries are from vector $\textbf{x}$. $x_i$ denots the $i$th entry of $\textbf{x}$, $x_{i,j}$ is the element in $i$-th row and $j$-th column of a matrix $\textbf{X}$, $\textbf{x}_i$ denots the $i$-th column of matrix $\textbf{X}$. $\text{tr}\left(\textbf{X}\right)$ denotes the trace of $\textbf{X}$, $\text{rank}\left(\textbf{X}\right)$ represents the rank of $\textbf{X}$, $\text{vec}\left(\textbf{X}\right)$ means the vectorization of matrix $\textbf{X}$, which is a linear transformation which converts $\textbf{X}$ into a column vector, $\lambda_{\max} \left(\textbf{X}\right)$ is the maximum eigenvalue of matrix $\textbf{X}$. $\textbf{X} \ge 0$ means that matrix $\textbf{X}$ is positive semidefinite. $\odot$ stands for the Hadamard product, $\otimes$ denotes Kronecker product. $\left(\cdot\right)^H$ is the conjugate transpose; $\left(\cdot\right)^C$ represents the conjugate and $\left(\cdot\right)^T$ denotes the transpose. $\mathbb{C}$ stands for the complex domain while $\mathbb{R}$ represents the real domain. $\mathcal{CN}\left(\pmb{\mu}, \pmb{\Sigma}\right)$ stands for the circularly symmetric complex Gaussian distribution with mean $\pmb{\mu}$ and covariance $\pmb{\Sigma}$. $E\left[ \cdot\right]$ represents the expectation operation. $\textbf{I}_N$ denotes an $N\times N$ identity matrix. $\left[x\right]^+$ denotes $\rm{max}\left(x, 0\right)$. ${\mathop{\rm Re}\nolimits} {{\left\{ x \right\}}}$ and ${\mathop{\rm Im}\nolimits} {{\left\{ x \right\}}}$ represent the real and imaginary part of $x$, respectively. $\nabla$ denotes the gradient of a function. $Q(\cdot)$ stands for the tail distribution function of the standard normal distribution.


\section{ System Model and Problem Formulation }
\begin{figure}[t]
\centering
\includegraphics[width= 3.4in,angle=0]{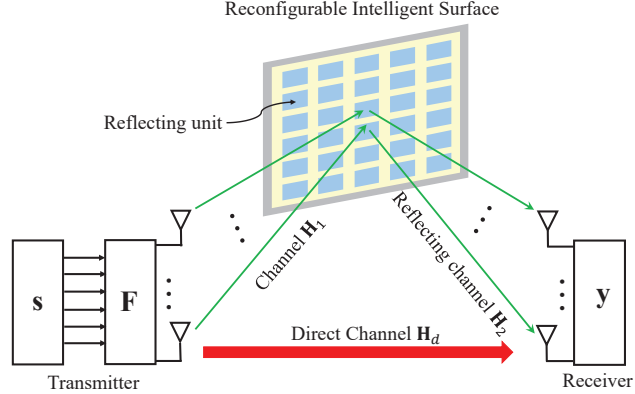}
\caption{System model.}
\label{fig_1}
\end{figure}
In this paper, we consider an RIS-assisted MIMO system model as illustrated in Fig. 1. In the model, a transmitter equipped with $N_t$ antennas communicates with a receiver equipped with $N_r$ antennas with the help of a RIS composed of $N$ reflecting units. The RIS acts as a passive relay, which is embedded on a surrounding building. 
The received vector $ \textbf{y} \in {\mathbb{C}^{N_r \times 1}}$ at the receiver can be expressed as 

\begin{align}
 \textbf{y} =  \sqrt{\rho} \left(\textbf{H}_2 \pmb{\Phi} \textbf{H}_1 +\textbf{H}_d\right)\textbf{F} \textbf{s} + \textbf{n},
\end{align}where $\rho$ is the SNR; $\textbf{H}_2 \in {\mathbb{C}^{N_r \times {N}}}$ represents the channel between the RIS and the receiver; $\pmb{\Phi} = \text{diag}\left\{ \pmb{\phi} \right\}\in {\mathbb{C}^{N \times {N}}}$ denotes the diagonal matrix accounting for the effective phase shifts applied by the RIS reflecting elements with $\pmb{\phi} {\rm{ =  }}\left[ \phi_1, \phi_2, \cdots, \phi_N \right]^T$, where $\left|\phi_{n}\right| = 1$ with $\phi_{n} \in \mathcal{F} \triangleq\left\{\exp \left(\frac{{\rm{j}} 2 \pi m}{2^{b}}\right)\right\}_{m=0}^{2^{b}-1}$ with $b$ phase resolution in number of bits \cite{CH1}; $\textbf{H}_1 \in {\mathbb{C}^{N \times {N_t}}}$ represents the channel between the transmitter and the RIS, $\textbf{H}_d \in {\mathbb{C}^{N_r \times {N_t}}}$ represents the direct channel between the transmitter and the receiver; $\textbf{F} \in {\mathbb{C}^{N_t \times {N_s}}}$ is the precoder to encode $N_s$ data streams; $\textbf{s}$ is the ${N_s \times 1}$ transmitted data symbol vector with each entry chosen from a $M$-ary constellation $\mathcal{S}_M$ and there are totally $M^{N_s}$ legitimate symbol vectors; $ \textbf{n} \sim \mathcal{CN}\left( {\textbf{0},\sigma^2{\textbf{I}_{{N_r}}}} \right)$ is the additive white Gaussian noise (AWGN) vector with each entry obeying a zero-mean variance $\sigma^2$ complex Gaussian distribution. In this paper, we assume that the average power of all legitimate symbol vectors $\{\textbf{s}\}$ is normalized. Let $\textbf{x} = \textbf{F} \textbf{s}$ denote the transmitted signal vector from the multi-antenna transmitter. It is assumed that $\textbf{x}$ satisfies the maximum average transmit power constraint $E\left[ {\left\| \textbf{x} \right\|_2^2} \right] = E\left[\text{tr}\left( {\textbf{F}\textbf{s}{\textbf{s}^H}{\textbf{F}^H}} \right)\right] = \text{tr}\left( {\textbf{Q}} \right) \le {P_{\max }}$, where $\textbf{Q}\mathop {\rm{ = }}\limits^\Delta  \textbf{F}{\textbf{F}^H}$ is the signal covariance matrix and ${P_{\max }}$ denotes the maximum average power.

It is considered that the system utilizes all $M^{N_s}$ feasible transmit vectors, the union bound on SER can thus be written as
\begin{align}\label{ob}
{P_S}\left( {\pmb{\Phi} ,\textbf{F}} \right) &= \frac{1}{{M^{N_s}}}\sum\limits_{i = 1}^{M^{N_s}} {\sum\limits_{j = 1,j \ne i}^{M^{N_s}} \Pr \left\{ {{\textbf{s}_i} \to {\textbf{s}_j}} \right\}}, 
\end{align}
where $\Pr \left\{ {{\textbf{s}_i} \to {\textbf{s}_j}} \right\}$ denotes the pairwise SER of the vector ${\textbf{s}_i}$ being erroneously detected as ${\textbf{s}_j}$. By using the squared Euclidean distance $d_{ij}^2\left( {\pmb{\Phi} ,\textbf{F}} \right) = {\left\| {\left({\textbf{H}_2}\pmb{\Phi} {\textbf{H}_1}+{\textbf{H}_d}\right) \textbf{F}\left( {{\textbf{s}_i} - {\textbf{s}_j}} \right)} \right\|^2}$ between two vectors,  $\Pr \left\{ {{\textbf{s}_i} \to {\textbf{s}_j}} \right\}$ can be computed as $\Pr \left\{ {{\textbf{s}_i} \to {\textbf{s}_j}} \right\} = {Q\left( {\sqrt {\frac{\rho{d_{ij}^2\left( {\pmb{\Phi} ,\textbf{F}} \right)}}{2\sigma^2}} } \right)}$. 

The objective of our design is to minimize the SER. To make the targeted problem more tractable, we assume that the phase resolution is infinite \cite{CH} and all involved CSI are known, the optimization can be formulated as problem $\left(\textbf{P1}\right)$:
\begin{equation}
\begin{split}
 \left(\textbf{P1}\right): \;\;\;\rm{Given}: &~\textbf{H}_1, \textbf{H}_2, \textbf{H}_d, \mathcal{S}_M \\
\rm{Find}: &~\pmb{\Phi}, \textbf{F}\\
\rm{Minimize}: &~{P_S}\left( {\pmb{\Phi} ,\textbf{F}} \right)\\
\rm{Subject\;to}: &~\text{tr}\left( {\textbf{Q}} \right) \le {P_{\max }}\\
&~\left| {{\phi _i}} \right| = 1, \forall i = 1,...,N,
\end{split}
\end{equation}
where the first constraint ensures that the BS transmit power is kept below the maximum feasible power $P_{\max}$ and the second constraint means that every reflecting unit only provides phase shift without signal amplification.

The joint optimization problem is non-convex and challenging to obtain the optimal solution due to the coupling effect directly. Thus, we resort to an alternating way to optimize the reflecting and precoding, similarly to \cite{CP1, CP,HG,QWW}.

\section{Joint Reflecting and Precoding Design}
As stated above, because of the coupling effect between the reflecting elements in $\pmb{\Phi}$ and the precoder $\textbf{F}$, the original problem $\left(\textbf{P1}\right)$ is hard to solve. To decouple them, we will first optimize $\pmb{\Phi}$ by fixing $\textbf{F}$ and then update $\textbf{F}$ by fixing $\pmb{\Phi}$ respectively. Then, we will obtain sub-optimal solutions for both $\pmb{\Phi}$ and $\textbf{F}$ by performing the process iteratively until reaching the convergence of the objective function or the solutions to an acceptable level. In the following, we first present the designs for reflecting elements.

\subsection{Reflecting Designs}
Based on a given $\textbf{F}$, the squared Euclidian distance can be re-expressed as $
d_{ij}^2\left( {\pmb{\Phi}} \right) = \left\| {{\textbf{H}_2}\pmb{\Phi} {\textbf{H}_1}{{\textbf{x}_{ij}}}} \right\|_2^2+ 2\mathcal{R}\left(\pmb{\phi}^T{\textbf{a}_{ij}}\right) +\left\| {{\textbf{H}_d}{{\textbf{x}_{ij}}}} \right\|_2^2$, where ${\textbf{x}_{ij}} = \textbf{F}\left( {{\textbf{s}_i} - {\textbf{s}_j}} \right)$ and ${\textbf{a}_{ij}}$ is a vector with $k$-th element ${\textbf{a}_{ij,k}} = {{\textbf{x}_{ij}^H}{\textbf{h}_{d,2,k}}\textbf{h}_{1,k}^T{\textbf{x}_{ij}}}$, where ${\textbf{h}_{d,2,k}}$ is the $k$-th column of ${\textbf{H}_d^H}{{\textbf{H}_2}}$ and $\textbf{h}_{1,k}$ denotes the $k$-th column of $\textbf{H}_{1}^T$. 

Then, we carry out some manipulations to transform $\left\| {{\textbf{H}_2}\pmb{\Phi} {\textbf{H}_1}{{\textbf{x}_{ij}}}} \right\|_2^2$ as
\begin{align}\label{sEd}
{\left\| {{\textbf{H}_2}\pmb{\Phi} {\textbf{H}_1}{{\textbf{x}_{ij}}}} \right\|_2^2}&={\textbf{q}_{ij}^H}\pmb{\Phi}^H{\textbf{H}_2}^H{\textbf{H}_2}\pmb{\Phi}{\textbf{q}_{ij}}= {\rm{tr}}\left(\pmb{\Phi}^H{\textbf{R}_{\textbf{H}2}}\pmb{\Phi}\Delta {\textbf{Q}_{ij}}\right),
\end{align}
where ${\textbf{q}_{ij}} = {\textbf{H}_1}{{\textbf{x}_{ij}}}$ $\textbf{R}_{\textbf{H}2} = {\textbf{H}_2}^H{\textbf{H}_2}$ and $\Delta {\textbf{Q}_{ij}} = {{\textbf{q}_{ij}}}{{\textbf{q}_{ij}^H}}$. Using the rule that $\rm{tr}\left( {\textbf{D}_\textbf{x}^H\textbf{A}{\textbf{D}_\textbf{y}}{\textbf{B}^T}} \right) = {\textbf{x}^H}\left( {\textbf{A} \odot \textbf{B}} \right)\textbf{y}$ with $\textbf{D}_\textbf{x} = \rm{diag}\left\{\textbf{x}\right\}$ and $\textbf{D}_\textbf{y} = \rm{diag}\left\{\textbf{y}\right\}$ in \cite{XZ}, we can re-express ${\left\| {{\textbf{H}_2}\pmb{\Phi} {\textbf{H}_1}{{\textbf{x}_{ij}}}} \right\|_2^2}$ as ${\left\| {{\textbf{H}_2}\pmb{\Phi} {\textbf{H}_1}{{\textbf{x}_{ij}}}} \right\|_2^2} =\pmb{\phi}^H\left(\textbf{R}_{\textbf{H}2}\odot\Delta {\textbf{Q}_{ij}^T}\right)\pmb{\phi}= \pmb{\phi}^H \textbf{C}_{ij}\pmb{\phi}$, where $\textbf{C}_{ij} = \textbf{R}_{\textbf{H}2}\odot\Delta {\textbf{Q}_{ij}^T}$. 

Finally, $d_{ij}^2\left( {\pmb{\Phi} }\right)$ can be finally expressed by the function of $\pmb{\phi}$ as
\begin{align}\label{dis}
d_{ij}^2\left( {\pmb{\phi}} \right) = \pmb{\phi}^H \textbf{C}_{ij}\pmb{\phi} + 2\mathcal{R}\left(\pmb{\phi}^T{\textbf{a}_{ij}}\right) +\left\| {{\textbf{H}_d}{{\textbf{x}_{ij}}}} \right\|_2^2.
\end{align}

\subsubsection{eMSER-Refelcting Scheme}
Expanding the matrix product, we can re-express the minimum Euclidian distance as $d_{ij}^2\left( {\pmb{\phi}} \right) = \sum\limits_{n = 1}^N\sum\limits_{k = 1}^N\phi _n{\phi _k}^H C_{ij,k,n}+2\mathcal{R}\left(\sum\limits_{k = 1}^N{\phi _k} a_{ij,k}\right)+\left\| {{\textbf{H}_d}{{\textbf{x}_{ij}}}} \right\|_2^2$,
where $C_{ij,k,n}$ is the $k$-th row and $n$-th column of matrix $\textbf{C}_{ij}$. In order to handle the non-convex problem, we use a coordinate descent method to obtain the sub-optimal solution.
 Recall $\phi_n = \exp{\left({\rm{j}}\theta_n\right)}$, the objective can be represented as a function of ${\theta_n}$ by fixing remaining phase shifts. In details, we optimize $\phi_n$ by fixing remaining phase shifts firstly until it meets the stop criterion. Then, we do the process iteratively until we get the sub-optimal solution. Thus, we formulate the gradient of the cost function $P_S\left({\theta_n}\right)$ over ${\theta_n}$ as
\begin{align}\label{theta}
\frac{\partial  P_S\left({\theta_n}\right)}{\partial \theta_n} = \frac{1}{{M}^{N_s}}\sum\limits_{i = 1}^{M^{N_s}} {\sum\limits_{j = 1,j \ne i}^{M^{N_s}} \frac{\partial{Q\left( {\sqrt {\frac{\rho {d_{ij}^2\left( \theta_n \right)}}{2\sigma^2}} } \right)}}{\partial \theta_n}}.
\end{align}

According to Leibniz's integral rule, we can express $\frac{\partial{Q\left( {\sqrt {\frac{\rho {d_{ij}^2\left( \theta_n \right)}}{2\sigma^2}} } \right)}}{\partial \theta_n}$ as
\begin{align}
&\frac{\partial{Q\left( {\sqrt {\frac{\rho {d_{ij}^2\left( \theta_n \right)}}{2\sigma^2}} } \right)}}{\partial \theta_n}=-\sqrt{\frac{\rho}{\pi\sigma^2 d_{ij}^2 \left( { \theta_n} \right)}}\exp \left( { - \frac{\rho {d_{ij}^2\left( { \theta_n} \right)}}{4\sigma^2}} \right)\frac{\partial d_{ij}^2\left( \theta_n \right)}{\partial \theta_n},
\end{align}
where $\frac{\partial d_{ij}^2\left( \theta_n \right)}{\partial \theta_n}  = 2\mathcal{R}\left(\sum\limits_{k = 1,k\neq n}^N{\rm{j}}\exp\left\{{\rm{j}}\left(\theta_n-\theta_k\right)\right\}C_{ij,k,n}\right)+2\mathcal{R}\left({\rm{j}}\exp\left\{{\rm{j}}\theta_n\right\} a_{ij,n}\right)$.

By using $-\frac{\partial  P_S\left({\theta_n}\right)}{\partial \theta_n} $ as the search direction \cite{WW} for ${\theta_n}$ and do the process iteratively, we can search the optimized solution as listed in Algorithm 1. Since we optimize $\theta_n$ directly, the optimization process will not lead to an infeasible solution. This is because for any real number $\theta_n$, $|e^{\textrm{j}\theta_n}|=1$ always holds.

\begin{algorithm}[t]
\caption{eMSER-Reflecting Scheme}
\label{10}
\begin{algorithmic}[1]
\STATE {\textbf{Initialization:} Given a feasible initial solution ${\pmb{\theta}}_0$, $k = 0$, $n = 1$, halting criterion ${\varepsilon _0}> 0$ and halting criterion ${\varepsilon _1}> 0$. }
\STATE {\textbf{Gradient and Search direction:} Compute the gradient $ \textbf{g}_k$ and derive the search direction as $w_{n,k} = -g_{n,k} = -\frac{\partial  P_S\left({\theta_{n,k}}\right)}{\partial \theta_{n,k}}$, where $-{\nabla _{{\theta_n}} }P_S\left(\theta_{n,k}\right)$ is given in \eqref{theta}. }
\STATE {\textbf{Update:} Choose Armijo backtracking line search step size $\alpha_{n,k}$ and
\begin{align}\label{it}
\theta_{n,k+1} = \theta_{n,k}+\alpha_{n,k} w_{n,k},
\end{align}
$k \leftarrow k + 1$. If $\left|w_{n,k}\right| < {\varepsilon _0}$ and $n<N$, then $n \leftarrow n +1$ and go to step 2.}
\STATE {\textbf{Iteration:} Let $n = 1$ and go to Step 2 until $P_S\left(\theta_{n,k+1}\right)- P_S\left(\theta_{n,k}\right) < {\varepsilon _1}$.}
\STATE {\textbf{Output:} The optimized reflecting elements are thus given by ${{\phi}_n}^* = {\exp}\left[{{\rm{j}}}\theta_{n,k}\right]$.
}
\end{algorithmic}
\end{algorithm}

\subsubsection{vMMSR-Reflecting Scheme}
To facilitate practical implementation, we investigate low-complexity phase shift designs in this part. Since $\left| {{\phi _i}} \right| = 1$, accordingly we can obtain that $\text{tr}\left( {\pmb{\phi}{\pmb{\phi}}^H} \right) = N$. In order to handle the non-convex constraint of $\left| {{\phi _i}} \right| = 1$, we relax the problem $\left(\textbf{P1}\right)$ into the following optimization  $\left(\textbf{P2}\right)$ with a convex ${\ell _\infty }$ constraint:
\begin{equation}
\begin{split}
\left(\textbf{P2}\right): \;\;\;\rm{Given}: &~\textbf{H}_1, \textbf{H}_2, \textbf{H}_d, \mathcal{S}_M, \textbf{F}\\
\rm{Find}: &~\pmb{\phi} \\
\rm{Minimize}: &~{P_S}\left( {\pmb{\phi}} \right)\\
\rm{Subject\;to}:&~\text{tr}\left( {\pmb{\phi}{\pmb{\phi}}^H} \right) = N\\
 &~{\left\| \pmb{\phi}  \right\|_\infty } \le 1.
\end{split}
\end{equation}

Actually, the original feasible set is a subset of the new feasible set in $\left(\textbf{P2}\right)$, i.e.,
\begin{align}
&\left\{ {\pmb{\phi} \in {\mathbb{C}^{N \times {1}}}:\left| {{\phi _i}} \right| = 1,  \forall i = 1,...,N} \right\}= \left\{ {\pmb{\phi} \in {\mathbb{C}^{N \times {1}}}:\text{tr}\left( {\pmb{\phi}{\pmb{\phi}}^H} \right) = N \& {\left\| \pmb{\phi}  \right\|_\infty } \le 1 } \right\},
\end{align}
which is convex due to the convexity of ${\ell _\infty }$ norm.

Since the ${\ell _\infty }$ constraint is non-differentiable, we exploit the ${\ell _p }$ approximation \cite{LH} with a gradually increased large $p$, $\mathop {\lim }\limits_{p \to \infty } {\left\| \pmb{\phi} \right\|_p} = {\left\| \pmb{\phi} \right\|_\infty }$, during the optimization process. To solve $\left(\textbf{P2}\right)$, we utilize the barrier method to incorporate the non-negative constraint \cite{SP} with the logarithmic barrier function $I\left( u \right)$ to approximate the penalty of violating the ${\ell _p }$ constraint, i.e., $
I\left( u \right) = \left\{ {\begin{array}{*{20}{c}}
{ - \frac{1}{t}\ln \left( u \right),}&{u > 0}\\
{\infty ,}&{u \le 0,}
\end{array}} \right.$, where $t$ is used to scale the barrier function's penalty. Thus, we can obtain the following optimization problem:
\begin{align}\label{f}
\mathop {\min }\limits_{\pmb{\phi} \in {\mathbb{C}^{N \times {1}}}}  g\left( {\pmb{\phi}, p} \right) =  {P_S}\left( \pmb{\phi}  \right) + I \left( { 1-{{\left\| \pmb{\phi}  \right\|}_p}} \right).
\end{align}
To solve \eqref{f} via a gradient method, we formulate the gradient of the cost function $g\left( {\pmb{\phi}, p} \right)$ over $\pmb{\phi}$ as follows:
\begin{align}\label{g}
{\nabla _{\pmb{\phi}} }g\left( {\pmb{\phi}, p} \right) &= {\nabla _{\pmb{\phi}} }{P_S}\left( \pmb{\phi}  \right) + \frac{{\left\| \pmb{\phi}  \right\|_p^{1 - p}{\textbf{p}_{\pmb{\phi}} }}}{2t\left( { 1-{{\left\| \pmb{\phi}  \right\|}_p}} \right)},
\end{align}
where ${\textbf{p}_{\pmb{\phi}} }\in {\mathbb{C}^{N \times {1}}}$ is given as
${\textbf{p}_{\pmb{\phi}} } = {\left[ {{\phi _1} \cdot {{\left| {{\phi _1}} \right|}^{p - 2}},{\phi _2} \cdot {{\left| {{\phi _2}} \right|}^{p - 2}},...,{\phi _N} \cdot {{\left| {{\phi _N}} \right|}^{p - 2}}} \right]^T}$.

Moreover, the gradient ${\nabla _{\pmb{\phi}} }{P_S}\left( \pmb{\phi}  \right)$ can be calculated as
\begin{align}
{\nabla _{\pmb{\phi}} }{P_S}\left( \pmb{\phi}  \right) = \frac{1}{{M}^{N_s}}\sum\limits_{i = 1}^{M^{N_s}} {\sum\limits_{j = 1,j \ne i}^{M^{N_s}} {\nabla _{\pmb{\phi}} }{Q\left( {\sqrt {\frac{\rho {d_{ij}^2\left( {\pmb{\phi}} \right)}}{2\sigma^2}} } \right)}}.
\end{align}

According Leibniz's integral rule, $
{\nabla _{\pmb{\phi}} }{Q\left( {\sqrt {\frac{\rho {d_{ij}^2\left( {\pmb{\phi}} \right)}}{2\sigma^2}} } \right)}=-\sqrt{\frac{\rho}{\pi\sigma^2 d_{ij}^2 \left( {\pmb{\phi}} \right)}}\exp \left( { - \frac{\rho {d_{ij}^2\left( {\pmb{\phi}} \right)}}{4\sigma^2}} \right)\left(\textbf{C}_{ij}\pmb{\phi}+\textbf{a}_{ij}^C\right)$. Thus, ${\nabla _{\pmb{\phi}} }g\left( {\pmb{\phi}, p} \right)$ can be re-expressed as
\begin{align}\label{gra}
&{\nabla _{\pmb{\phi}} }g\left( {\pmb{\phi}, p} \right) = -\frac{1}{{M}^{N_s}}\sum\limits_{i = 1}^{M^{N_s}} {\sum\limits_{j = 1,j \ne i}^{M^{N_s}}}\sqrt{\frac{\rho}{\pi\sigma^2 d_{ij}^2 \left( {\pmb{\phi}} \right)}} \exp \left( { - \frac{\rho{d_{ij}^2\left( {\pmb{\phi}} \right)}}{4\sigma^2}} \right)\left(\textbf{C}_{ij}\pmb{\phi}+\textbf{a}_{ij}^C\right)+\frac{{\left\| \pmb{\phi}  \right\|_p^{1 - p}{\textbf{p}_{\pmb{\phi}} }}}{2t\left( { 1-{{\left\| \pmb{\phi}  \right\|}_p}} \right)}.
\end{align}

By using $-{\nabla _{\pmb{\phi}} }g\left( {\pmb{\phi}, p} \right) $ as the search direction \cite{WW}, we can search the optimized solution as listed in Algorithm 2.

\begin{algorithm}[t]
\caption{vMSER-Reflecting Scheme}
\label{10}
\begin{algorithmic}[1]
\STATE {\textbf{Initialization:} Given a feasible initial solution ${\pmb{\phi}}_0$, $p>0$, $\Delta p>0$, $p_{\rm{max}}>0$, $k = 0$, halting criterion ${\varepsilon _2}> 0$ and ${\varepsilon _3}> 0$ and the barrier coefficient $t$. }
\STATE {\textbf{Gradient and Search direction:} Compute the gradient $ \textbf{g}_k$ and derive the search direction as $
  \textbf{w}_k = -\textbf{g}_k = -{\nabla _{\pmb{\phi}} }g\left( {\pmb{\phi}_k, p} \right)$, where ${\nabla _{\pmb{\phi}} }g\left( {\pmb{\phi}_k, p} \right)$ is given in \eqref{gra}. }
\STATE {\textbf{Direction Projection:} Project the search direction into the tangent plane of $\text{tr}\left( {\pmb{\phi}{\pmb{\phi}}^H} \right) = N$ through $\textbf{w}_k^ \bot  = {\textbf{w}_k} - \frac{{\left\langle {{\textbf{w}_k}, {\pmb{\phi}}_k} \right\rangle }{\pmb{\phi}}_k}{{\left\| {{\pmb{\phi}}_k} \right\|^2}}$. }
\STATE {\textbf{Search for ${\hat \theta }$:} For $ 0 \le {\theta } \le \pi/2$, searching for it by $
{\hat \theta }  = \mathop {\arg \min }\limits_{\theta } {P_S}\left( {\pmb{\phi}}_{k}  \right)$.}
\STATE {\textbf{Update and Normalize:} Go to step 6 if $\frac{\textbf{w}_k^ \bot}{{\left\| {\textbf{w}_k } \right\|}} \le {\varepsilon _3}$ and $\left|{P_S}\left( {\hat{\pmb{\phi}}}_{k} \right)-{P_S}\left( {\pmb{\phi}}_{k}  \right)\right| < {\varepsilon _2}$, where the $i$-th entry of $\hat{\pmb{\phi}}_{k}$ is ${\hat{\phi}}_{k,i} = \frac{{\phi}_{k,i}}{\left|{\phi}_{k,i}\right|}$, else let 
\begin{align}\label{it}
{\pmb{\phi}}_{k+1} = \cos \hat \theta  \cdot {\pmb{\phi}}_{k}+\sin \hat \theta  \cdot \sqrt{N}\frac{\textbf{w}_k^ \bot}{{\left\| {\textbf{w}_k^ \bot } \right\|}},
\end{align}
$k \leftarrow k + 1$ and then go to step 2.
}
\STATE {\textbf{Iteration:} Go to Step 7 if $p \ge p_{\rm{max}}$, else let $p \leftarrow p + \Delta p$ and then go to Step 2.}
\STATE {\textbf{Output:} The optimized reflecting elements are thus given by $
{\pmb{\phi}}^{*} = \hat{\pmb{\phi}}_k$.
}
\end{algorithmic}
\end{algorithm}

\begin{proposition}
\rm{In each iteration of Algorithm 2, ${P_S}\left( {\pmb{\phi}}_{k+1}  \right) \le {P_S}\left( {\pmb{\phi}}_{k}  \right)$}.
\end{proposition}
\begin{proof}
For any ${\hat \theta }\rightarrow 0$, based on \eqref{it}, the Taylor expansion of ${P_S}\left( {\pmb{\phi}}_{k+1}\right)$ can be derived as
\begin{align}\label{Ty}
{P_S}\left( {\pmb{\phi}}_{k+1}  \right) &= {P_S}\left( {\pmb{\phi}}_{k} \right)+{\textbf{g}_k^H}\cdot \sqrt{N}\frac{\textbf{w}_k^ \bot}{{\left\| {\textbf{w}_k^ \bot } \right\|}} + \mathcal{O}\left({\hat \theta }^2\right)\notag\\
&\approx {P_S}\left( {\pmb{\phi}}_{k} \right)+{\textbf{g}_k^H}{\textbf{w}_k^ \bot} \frac{\sqrt{N}}{{\left\| {\textbf{w}_k^ \bot } \right\|}}. 
\end{align}
We can also calculate that ${\textbf{g}_k^H}{\textbf{w}_k^ \bot} = \left(\cos^2\alpha-1\right){{\left\| {\textbf{w}_k } \right\|^2}} \le 0$, where $\alpha = \arccos \frac{{\pmb{\phi}}_{k}^H\textbf{w}_k}{{\left\|{\pmb{\phi}}_{k}^H {\textbf{w}_k} \right\|}}$ is the angle between vectors ${\pmb{\phi}}_{k}$ and ${\textbf{w}_k}$. Combining it with \eqref{Ty}, we have ${P_S}\left( {\pmb{\phi}}_{k+1}  \right) \le {P_S}\left( {\pmb{\phi}}_{k}  \right)$.
The proof is completed.
\end{proof}

Under this constraint relaxation, the solution obtained from each iteration may cannot guarantee satisfying $\left| {{\phi _{k,i}}} \right| = 1$. To fix this, we add a element-wise normalization step $\hat{\pmb{\phi}}_{k}$ is ${\hat{\phi}}_{k,i} = \frac{{\phi}_{k,i}}{\left|{\phi}_{k,i}\right|}$, conditioned that the objective function is not much sensitive to the normalization. This step ensures the feasibility of the solution and meanwhile keeps good performance.

\subsection{Precoding Designs}
Based on the optimized $\pmb{\Phi}$, we can simplify $\left(\textbf{P1}\right)$ as
\begin{equation}
\begin{split}
\left(\textbf{P3}\right): \;\;\;\rm{Given}: &~\textbf{H}_1, \textbf{H}_2, \textbf{H}_d, \mathcal{S}_M, \pmb{\Phi}\\
\rm{Find}: &~\textbf{F}\\
\rm{Minimize}: &~{P_S}\left( {\textbf{F}} \right)\\
\rm{Subject\;to}: &~\text{tr}\left( {\textbf{Q}} \right) \le {P_{\max }}.    
\end{split}
\end{equation}

By introducing  ${{\textbf{H}}} = {\textbf{H}_2}\pmb{\Phi}{\textbf{H}_1}+{\textbf{H}_d}  \in {\mathbb{C}^{N_r \times {N_t}}}$, the received signal can be re-expressed as
\begin{align}\label{Y}
 \textbf{y} &= {{\textbf{H}}} \textbf{F} \textbf{s} + \textbf{n}=\left[ {{{ \textbf{h}}_1},\cdots,{{ \textbf{h}}_{N_t}}} \right]\left[ {\begin{array}{*{20}{c}}
{{f_{1,1}}}&\cdots&{{f_{1,{N_s}}}}\\
\vdots&\ddots&\vdots\\
{{f_{{N_t},1}}}&\cdots&{{f_{{N_t},{N_s}}}}
\end{array}} \right]\left[ {\begin{array}{*{20}{c}}
{{s_0}}\\
{\vdots}\\
{{s_1}}
\end{array}} \right] + \textbf{n}=\sum\limits_{a = 1}^{{N_t}} {\sum\limits_{b = 1}^{{N_s}} {{f_{a,b}}{{ \textbf{h}}_{a}}{s_b}} } + \textbf{n}.
\end{align}

Considering that \eqref{Y} has two summations, which are not easy to handle, we rebuild the channel matrix, precoding matrix and the transmitted data stream in a new form \cite{PC}. We first construct a new channel matrix as ${{\bm{\hat} \textbf{H}}} = \left[ {{\bm{\hat} \textbf{H}}_1,\cdots,{\bm{\hat} \textbf{H}}_{n_t},\cdots,{\bm{\hat} \textbf{H}}_{N_t}} \right] \in {\mathbb{C}^{N_r \times {N_tN_s}}}$, in which $
{\bm{\hat} \textbf{H}}_{n_t} = \underbrace {\left[ {{{ \textbf{h}}_{n_t}},\cdots,{{ \textbf{h}}_{n_t}},\cdots,{{ \textbf{h}}_{n_t}}} \right]}_{{N_s}}
$, where each ${{\textbf{h}}_{n_t}}$ repeats $N_s$ times. Meanwhile, the precoding matrix entries in $\textbf{F}$ are collected together as $
{{\bm{\hat} \textbf{F}}} = \text{diag}\left\{\text{vec}\left(\textbf{F}^T\right)\right\}=\text{diag}\left\{\left[f_{1,1}, f_{1,2},\cdots,f_{N_t, N_s}\right]^T\right\}=\text{diag}\left\{\textbf{f}\right\}$. Furthermore,  ${{\bm{\hat} \textbf{s}}} = {\underbrace {\left[ {{{\textbf{s}}},\cdots,{{\textbf{s}}},\cdots,{{\textbf{s}}}} \right]}_{{N_t}}}^T\in {\mathbb{C}^{N_tN_s\times {1}}}$.

Following the procedure above, we can rewrite \eqref{Y} to be $\textbf{y} = {{\bm{\hat} \textbf{H}}}{{\bm{\hat} \textbf{F}}}{{\bm{\hat} \textbf{s}}} + \textbf{n}$. Thus, we can re-express the squared Euclidean distance as $d_{ij}^2\left( {\bm{\hat} \textbf{F}} \right) = {\left\| {{{\bm{\hat} \textbf{H}}} {\bm{\hat} \textbf{F}}\left( {{{\bm{\hat} \textbf{s}}_i} - {{\bm{\hat} \textbf{s}}_j}} \right)} \right\|^2}$. Similarly to \eqref{sEd}, the squared Euclidean distance can be re-expressed as
\begin{align}\label{df}
d_{ij}^2\left( {\bm{\hat} \textbf{F}} \right) &=\left({{{\bm{\hat} \textbf{s}}_i} - {{\bm{\hat} \textbf{s}}_j}}\right)^H{\bm{\hat} \textbf{F}}^H{\bm{\hat} \textbf{H}}^H{\bm{\hat} \textbf{H}}{\bm{\hat} \textbf{F}}\left({{{\bm{\hat} \textbf{s}}_i} - {{\bm{\hat} \textbf{s}}_j}}\right)= {\rm{tr}}\left({\bm{\hat} \textbf{F}}^H{\textbf{R}_{{\bm{\hat} \textbf{H}}}}{\bm{\hat} \textbf{F}}\Delta {\textbf{S}_{ij}}\right)=\textbf{f}^H\left({\textbf{R}_{{\bm{\hat} \textbf{H}}}}\odot\Delta {\textbf{S}_{ij}}\right)\textbf{f} = \textbf{f}^H {\bm{\hat} \textbf{C}}_{ij}\textbf{f},
\end{align}
where ${\textbf{R}_{{\bm{\hat} \textbf{H}}}} = {\bm{\hat} \textbf{H}}^H{\bm{\hat} \textbf{H}}$, $\Delta {{\textbf{S}}_{ij}} = \left( {{{\bm{\hat} \textbf{s}}_i} - {{\bm{\hat} \textbf{s}}_j}} \right)\left({{{\bm{\hat} \textbf{s}}_i} - {{\bm{\hat} \textbf{s}}_j}}\right)^H$, $\textbf{f} = \left[f_{1,1}, f_{1,2},\cdots,f_{N_t, N_s}\right]^T$ and ${\bm{\hat} \textbf{C}}_{ij} = {\textbf{R}_{{\bm{\hat} \textbf{H}}}}\odot\Delta {\textbf{S}_{ij}}$.

In the following, we will provide several algorithms to solve this problem.

\subsubsection{MMSE-Precoding Scheme}
 Based on discussion above, $\left(\textbf{P3}\right)$ can be transformed to 
\begin{equation}
\begin{split}
\left(\textbf{P3-\textbf{a}}\right): \;\;\;\rm{Given}: &~{\bm{\hat} \textbf{H}}, \mathcal{S}_M\\
\rm{Find}: &~\textbf{f}\\
\rm{Minimize}: &~{P_S}\left( {\textbf{f}} \right)\\
\rm{Subject\;to}: &~\text{tr}\left( {\textbf{f}\textbf{f}^H} \right) \le {P_{\max }}.
\end{split}
\end{equation}

To solve $\left(\textbf{P3-\textbf{a}}\right)$, we formulate the Lagrangian function as
\begin{align}
 L\left( {\textbf{f}, \mu} \right) =  {P_S}\left( \textbf{f} \right) + \mu\left(\text{tr}\left( {\textbf{f}\textbf{f}^H} \right) - {P_{\max }}\right),
\end{align}
where $\mu$ is the Lagrangian multiplier. The optimal solution must satisfy the Karush-Kuhn-Tucker (KKT) conditions as
\begin{align}\label{KKT}
\left\{ \begin{array}{l}
{\nabla _{\textbf{f}} } L\left( {\textbf{f}, \mu} \right) = 0\\
\mu\left(\text{tr}\left( {\textbf{f}\textbf{f}^H} \right) - {P_{\max }}\right) = 0\\
\mu \ge 0
\end{array} \right..
\end{align}
Because of its monotonicity with power, $ L\left( {\textbf{f}, \mu} \right)$ is minimized when the power constraint is met with strict equality. Hence, $\mu\left(\text{tr}\left( {\textbf{f}\textbf{f}^H} \right) - {P_{\max }}\right) = 0$. The first equation in \eqref{KKT} can be represented as
\begin{align}\label{LL}
{\nabla _{\textbf{f}} } L\left( {\textbf{f}, \mu} \right) &= \left[-\frac{1}{{M^{N_s}}}\pmb{\Omega}\left(\textbf{f}\right)+2\mu\textbf{I}\right]{\textbf{f}} = 0,
\end{align}
where $\pmb{\Omega}\left(\textbf{f}\right)$ is given by
\begin{align}\label{ff}
\pmb{\Omega}\left(\textbf{f}\right) = \sum\limits_{i = 1}^{M^{N_s}} {\sum\limits_{j = 1,j \ne i}^{M^{N_s}}}{ {\sqrt {\frac{\rho}{4\pi\sigma^2 {d_{ij}^2\left( {\textbf{f}} \right)}}} }}\exp \left( { - \frac{\rho{d_{ij}^2\left( {\textbf{f}} \right)}}{4\sigma^2}} \right){\bm{\hat} \textbf{C}}_{ij}.
\end{align}
Clearly, the closed-form solution of \eqref{LL} is difficult to derive. In this paper, we use the gradient method and project the solution onto the tangent plane of $\text{tr}\left( {\textbf{f}\textbf{f}^H} \right) = {P_{\max }}$. Specifically, the gradient descent direction of \eqref{ff} is firstly calculated to be
\begin{align}\label{sd}
  \textbf{r}_k = -{\bm{\hat} \textbf{g}}_{k} =\pmb{\Omega}\left(\textbf{f}_k\right)\textbf{f}_k.
\end{align}
Then, a projection on search direction is conducted by
\begin{align}\label{sdp}
\textbf{r}_k^ \bot  = {\textbf{r}_k} - \frac{{\left\langle {{\textbf{r}_k}, \textbf{f}_k} \right\rangle }\textbf{f}_k}{{\left\| {\textbf{f}_k} \right\|^2}}.
\end{align} 
We update the solution by searching along the projected search direction as
\begin{align}\label{fn}
\textbf{f}_{k+1} = \cos \hat \beta  \cdot {\textbf{f}}_{k}+\sin \hat \beta \cdot \sqrt{{P_{\max }}}\frac{\textbf{r}_k^ \bot}{{\left\| {\textbf{r}_k^ \bot } \right\|}},
\end{align}
where $\hat \beta $, $ \left(0 \le {\hat \beta } \le \pi/2\right)$ can be obtained by 
\begin{align}\label{beta}
{\hat \beta}  = \mathop {\arg \min }\limits_{\hat \beta } {P_S}\left( \textbf{f}_{k}  \right).
\end{align}
By doing this process iteratively until it reaches the stop criterion $\frac{\textbf{r}_k^ \bot}{{\left\| {\textbf{r}_k } \right\|}} \le {\varepsilon _4}$, where ${\varepsilon _4}$ is the halting criterion, a good solution can be obtained. For clearness, we list the iterative algorithm to search for the solution of \eqref{LL} in Algorithm 3.
\begin{algorithm}[t]
\caption{MSER-Precoding Scheme}
\begin{algorithmic}[1]
\STATE {\textbf{Initialization:} Given a feasible initial solution $\textbf{f}_0$, $k = 0$ and halting criterion ${\varepsilon _4}> 0$. }
\STATE {\textbf{Gradient and Search direction:} Compute the gradient $ \textbf{g}_k$ and derive the search direction $\textbf{r}_k$ through \eqref{sd} . }
\STATE {\textbf{Direction Projection:} Project the search direction into the tangent plane of $\text{tr}\left( {\textbf{f}\textbf{f}^H} \right) = {P_{\max }}$ through \eqref{sdp}. }
\STATE {\textbf{Search for ${\hat \beta }$:} For $ 0 \le {\hat \beta } \le \pi/2$, searching for it by \eqref{beta}.}
\STATE {\textbf{Update:} Go to step 6 if $\frac{\textbf{r}_k^ \bot}{{\left\| {\textbf{r}_k } \right\|}} \le {\varepsilon _4}$, else update the solution as \eqref{fn} and let $k \leftarrow k + 1$, and then go to step 2.
}
\STATE {\textbf{Output:} The optimized precoding matrix are thus given by:
\begin{align}
\textbf{F}^* = \left[ {\begin{array}{*{20}{c}}
{{f_{k1}}}&\cdots&{{f_{k{N_s}}}}\\
\vdots&\ddots&\vdots\\
{{f_{k\left( {{N_t} - 1} \right){N_s}}}}&\cdots&{{f_{k{N_t}{N_s}}}}
\end{array}} \right],
\end{align}
where ${{f_{ki}}}$ is the $i$-th elements of $\textbf{f}_k$.}
\end{algorithmic}
\end{algorithm} 

\begin{proposition}
\rm{In Algorithm 3, ${P_S}\left( \textbf{f}_{k+1}  \right) \le {P_S}\left( {\textbf{f}}_{k}  \right)$}.
\end{proposition}
\begin{proof}
It can be proved similarly to \textbf{Proposition 1}. 
\end{proof}

Since $\left(\textbf{P3-a}\right)$ is a non-convex problem, \eqref{LL} is a necessary condition for global optimum and the generated vector $\textbf{f}$ in Algorithm 3 is thus a critical point. This algorithm guarantee the feasibility for all solutions, because
\begin{align}
\textrm{tr}(\textbf{f}_{k+1}\textbf{f}_{k+1}^H)&=\textbf{f}_{k+1}^H\textbf{f}_{k+1}=\cos \hat \beta ^2\textbf{f}_{k}^H\textbf{f}_{k}+\sin \hat \beta^2P_{\max }+2\cos \hat \beta\sin \hat \beta\sqrt{{P_{\max }}}\textbf{f}_{k}^H \frac{\textbf{r}_k^ \bot}{{\left\| {\textbf{r}_k^ \bot } \right\|}}\notag \\
&= \cos \hat \beta ^2P_{\max }+\sin \hat \beta^2P_{\max }=P_{\max}.
\end{align}

\subsubsection{MMED-Precoding Scheme}In the high SNR regime, the SER for a given channel can be simplified as ${P_{S}}\left( {\pmb{\Phi} ,\textbf{F}} \right) \approx \frac{\lambda}{{M^{N_s}}}  {Q\left( {\sqrt {\frac{\rho{d_{\rm{min}}^2\left( {\pmb{\Phi} ,\textbf{F}} \right)}}{2\sigma^2}} } \right)}$, where $\lambda$ is the number of closest symbol pairs, and ${d_{\rm{min}}^2\left( {\pmb{\Phi} ,\textbf{F}} \right)= {\rm{min}}_{\forall {i,j,i \ne j}}d_{ij}^2\left( {\pmb{\Phi} ,\textbf{F}} \right)}$ is the minimum squared Euclidean distance between the noise-free received signal vectors. As 
${P_{S}}\left( {\pmb{\Phi} ,\textbf{F}} \right)$ is a monotonically decreasing function of ${d_{\rm{min}}^2\left( {\pmb{\Phi} ,\textbf{F}} \right)}$, problem $\left(\textbf{P1}\right)$ can be formulated as a maximizing the minimum Euclidian distance (MMED) problem \cite{shuai}.
Based on given $\pmb{\Phi}$ and \eqref{df}, $\left(\textbf{P1}\right)$ can be rewritten as 
\begin{equation}
\begin{split}
\left(\textbf{P4}\right): \;\;\;\rm{Given}: &~\textbf{H}_1, \textbf{H}_2, \textbf{H}_d, \mathcal{S}_M, {\pmb{\Phi}} \\
\rm{Find}: &~\textbf{f}\\
\rm{Maximize}: &~{\rm{min}}~ \textbf{f}^H {\bm{\hat} \textbf{C}}_{ij}\textbf{f}\\
\rm{Subject\;to}: &~\text{tr}\left( {\textbf{Q}} \right) \le {P_{\max }}.
\end{split}
\end{equation}

By introducing an auxiliary variable $r$, we have the equivalent epigraph of $\left(\textbf{P4}\right)$ as
\begin{equation}
\begin{split}
\left(\textbf{P4-a}\right): \;\;\;\rm{Given}: &~{\textbf{H}}, \mathcal{S}_M\\ 
\rm{Find}: &~\textbf{f}\\
\rm{Maximize}: &~r\\
\rm{Subject\;to}: &~\textbf{f}^H {\bm{\hat} \textbf{C}}_{ij}\textbf{f} \ge r\; {\forall {i,j,i \ne j}}\\
&~\text{tr}\left( {\textbf{f}\textbf{f}^H} \right) \le {P_{\max }}.
\end{split}
\end{equation}

However, we can review $\left(\textbf{P4-a}\right)$ in another equivalent way \cite{PC} as
\begin{equation}
\begin{split}
\left(\textbf{P4-b}\right): \;\;\;\rm{Given}: &{\textbf{H}}, \mathcal{S}_M\\ 
\rm{Find}: &~\textbf{f}\\
\rm{Minimize}: &~{\left\| \textbf{f}  \right\|^2 } \\
\rm{Subject\;to}: &~\textbf{f}^H {\bm{\hat} \textbf{C}}_{ij}\textbf{f} \ge d_{\rm{min}}\; {\forall {i,j,i \ne j}},
\end{split}
\end{equation}
where $d_{\rm{min}}$ is the desired squared minimum distance. The rationale behind $\left(\textbf{P4-b}\right)$ is to guarantee the minimum squared distance, while pursuing the minimum power usage as the objective. It can be seen that $\left(\textbf{P4-b}\right)$ is a large-scale non-convex QCQP problem. Since this optimization problem is similar to the problem investigated in \cite{PC}, we use the same method proposed in \cite{PC} to solve the problem. The details can be seen in \cite{PC}. It should be noted that the power constraint is first released in $\left(\textbf{P4-b}\right)$. After the optimization through MMED-Precoding, the obtained solution is scaled to satisfy the power constraint to ensure feasibility.

\subsection{Alternating Optimization}
Based on the aforementioned reflecting schemes and precoding schemes, an alternating algorithm can be conducted to minimize SER as listed in Algorithm 4. 

\begin{algorithm}[t]
\caption{Alternating Optimization}
\begin{algorithmic}[1]
\STATE {\textbf{Initialization:} Given a feasible initial solution ${\pmb{\phi}}_0$, $\textbf{f}_0$, $k = 0$ and stopping criteria ${\varepsilon_T}> 0$}
\STATE {\textbf{Optimize the reflecting elements:} Based on $\textbf{f}_k$, optimize the reflecting elements via reflecting schemes, which yields ${\pmb{\phi}}_{k+1}$.}
\STATE {\textbf{Optimize the precoder:} Based on ${\pmb{\phi}}_{k+1}$, optimize the precoder via precoding schemes, which yields $\textbf{f}_{k+1}$.}
\STATE {\textbf{Iteration:} Let $k \leftarrow k + 1$. Go to step 2 until ${P_S}\left( {\pmb{\phi}}_{k}, \textbf{f}_k  \right) -{P_S}\left( {\pmb{\phi}}_{k+1}, \textbf{f}_{k+1}  \right) < {\varepsilon_T}$.}
\end{algorithmic}
\end{algorithm}

\begin{proposition}
\rm{In Algorithm 4, ${P_S}\left({\pmb{\phi}}_{k+1}, \textbf{f}_{k+1}  \right) \le {P_S}\left({\pmb{\phi}}_{k}, {\textbf{f}}_{k}  \right)$}.
\end{proposition}
\begin{proof}
From \textbf{Proposition 1}, we can obtain that 
${P_S}\left({\pmb{\phi}}_{k+1}, \textbf{f}_{k}  \right) \le {P_S}\left({\pmb{\phi}}_{k}, {\textbf{f}}_{k}  \right)$ and from \textbf{Proposition 2}, we can further obtain that ${P_S}\left({\pmb{\phi}}_{k+1}, \textbf{f}_{k+1}  \right) \le {P_S}\left({\pmb{\phi}}_{k+1}, {\textbf{f}}_{k}  \right)$. Thus, in Algorithm 4, it yields
\begin{align}
{P_S}\left({\pmb{\phi}}_{k+1}, \textbf{f}_{k+1}  \right) \le {P_S}\left({\pmb{\phi}}_{k+1}, \textbf{f}_{k}  \right) \le {P_S}\left({\pmb{\phi}}_{k}, {\textbf{f}}_{k}  \right).
\end{align}
\end{proof}

 The proposition indicates that ${P_S}\left({\pmb{\phi}}_{k}, {\textbf{f}}_{k}  \right)$ is a monotone decreasing sequence as $k$ increases. As is well known, the sequence is bounded by a lower limit ${P_S}\left({\pmb{\phi}}^*, {\textbf{f}}^*  \right) \ge 0$, where $\left({\pmb{\phi}}^*, {\textbf{f}}^*  \right)$ represents the global optimal solution. According to the monotone convergence theorem, that is, if a sequence of real numbers is decreasing and bounded below, then its infimum is the limit, we can prove that our proposed alternating optimization algorithms are always converged. To observe the convergence speed, we have included a figure showing the number of iterations required for the algorithms in the simulation part. 

\begin{remark}
\rm{In this section, we split the optimization problem $\left(\textbf{P1}\right)$ into two separate optimization problem. The approach of iteratively solving the phase shift matrix and precoding matrix can only provide an efficient way to reduce the SER gradually. Due to the non-convexity of the problem, the approach only ensures a sub-optimal solution.  Moreover, the solution obtained through Alternating Optimization is also always feasible since the proposed precoding and reflecting optimization schemes are always feasible. }
\end{remark}

\section{Computational Complexity}
In this section, we analyze the computational complexity of the proposed algorithms.

\subsection{Computational Complexity of Reflecting Schemes}

\subsubsection{Complexity Order of eMSER-Reflecting} As can be seen from Algorithm 1, the computational complexity is dominated by the gradient calculation, which involves i) calculating the $M^{N_s}\left(M^{N_s}-1\right)$ matrix multiplications $\pmb{\phi}^H \textbf{C}_{ij}\pmb{\phi}$ for each angle optimization and ii) searching for the best solution for every angle based on other given angle, which is $N$ operations. Firstly, we should obtain the computational complexity of $\textbf{C}_{ij}$ computation, which is $\mathcal{C}_{C} =  \mathcal{O}\left[ N^2+N^2N_r+N_tN_s+NN_t+N^2 \right]$, where $N^2$ is the computation complexity for $\textbf{C}_{ij} = \textbf{R}_{\textbf{H}2}\odot\Delta {\textbf{Q}_{ij}^T}$ calculation, $N^2N_r$ is the computation complexity for $\textbf{R}_{\textbf{H}2}$ calculation and $N_tN_s+NN_t+N^2$ is the computation complexity for the calculation of $\Delta {\textbf{Q}_{ij}^T}$. It can be further simplified as $\mathcal{C}_{C} = \mathcal{O}\left[N^2N_r+N_tN_s+NN_t \right]$. Therefore, the complexity order of Algorithm 1 is :
\begin{equation}
\mathcal{C}_{\rm{eM}} = \mathcal{O}\left[ NM^{2N_s}\left(N^2N_r+N_tN_s+NN_t\right)\right].
\end{equation}
\subsubsection{Complexity of vMSER-Reflecting} As can be seen from Algorithm 2, the computational complexity is mainly consumed by the gradient calculation, which involves i) calculating the $M^{N_s}\left(M^{N_s}-1\right)$ matrix multiplications $\pmb{\phi}^H \textbf{C}_{ij}\pmb{\phi}$, and ii) calculating the ${\ell _p }$ norm. Therefore, the complexity order of vMSER-Reflecting Algorithm is :
\begin{align}
\mathcal{C}_{\rm{vM}}=\mathcal{O}\left[ M^{2N_s}\left(N^2N_r+N_tN_s+NN_t\right) + N\sum\limits_{n = 1}^{{N_p}} {{p^{\left( n \right)}}} \right],
\end{align}
where $N_p = \left(p_{\rm{max}}-p\right)/\Delta p$, $p^{\left( n \right)} = p + \left(n-1\right)\Delta p$, and $p$, $\Delta p$, $p_{\rm{max}}$ are specified in step 1 in Algorithm 2.

\subsection{Computational Complexity of Precoding Schemes}

\subsubsection{Complexity of MSER-Precoding} Similar to Algorithm 2, the complexity of Algorithm 3 is also mainly consumed by the gradient calculation involving matrices $\textbf{f}^H {\bm{\hat} \textbf{C}}_{ij}\textbf{f}$, while $\textbf{C}_{ij}$ needs $\mathcal{O}\left[N^2N_r\right.$ $\left.+N_rNN_t + N_rN_t^2N_s^2+N_t^2N_s^2\right] = \mathcal{O}\left[N^2N_r+N_rNN_t+ N_rN_t^2N_s^2\right]$, where $N^2N_r+N_rNN_t$ is the computation complexity of $\textbf{H} = \textbf{H}_2\pmb{\phi}\textbf{H}_1$, $N_rN_t^2N_s^2$ is the complexity for computing ${\textbf{R}_{{\bm{\hat} \textbf{H}}}} = {\bm{\hat} \textbf{H}}^H{\bm{\hat} \textbf{H}}$ and $N_t^2N_s^2$ is the complexity of $\Delta {\textbf{S}_{ij}} = \left( {{{\bm{\hat} \textbf{s}}_i} - {{\bm{\hat} \textbf{s}}_j}} \right)\left({{{\bm{\hat} \textbf{s}}_i} - {{\bm{\hat} \textbf{s}}_j}}\right)^H$. Therefore, the complexity order of Algorithm 3 for each iteration is :
\begin{equation}
\mathcal{C}_{\rm{MS}} = \mathcal{O}\left[ M^{2N_s}\left(N^2N_r+N_rNN_t+ N_rN_t^2N_s^2\right)\right].
\end{equation}

\subsubsection{Complexity of MMED-Precoding Scheme} 
According to \cite{PC}, the computational complexity order of MMED-Precoding $C_{\rm{MM}}$ is also mainly consumed by the gradient calculation involving matrices $\textbf{f}^H {\bm{\hat} \textbf{C}}_{ij}\textbf{f}$, which is same as Algorithm 2. 
\begin{remark}
\rm{It should be mentioned that MMED-Precoding does not need SNR information. In other words, MMED-Precoding is applicable for any SNR values, while MSER-Precoding needs to be re-designed as SNR varies. Besides, the number of iterations required for MMED-Precoding scheme is much less owing to the use of the BFGS method \cite{PC}. Thus, in general, the complexity of the MMED-Precoding scheme is much less than the MMSR-Precoding scheme.}
\end{remark}


\subsection{Complexity of Alternating Optimization} 
In this subsection, we analyze the overall complexity order of Alternating Optimization by taking different combinations of precoding schemes and reflecting schemes. By process reflecting schemes and three precoding schemes alternatively, the overall complexity order can be expressed as $\mathcal{C}=\mathcal{C}_{i \in \mathbb{R}}+\mathcal{C}_{k\in \mathbb{P}}$, 
where $\mathcal{C}_{i \in \mathbb{R}}$ and $\mathcal{C}_{k\in \mathbb{P}}$ denote the complexity order of the selected reflecting scheme and precoding scheme with $\mathbb{R} = \{\rm{eM}, \rm{vM}\}$ and $\mathbb{P} = \{\rm{MS}, \rm{MM}\}$.

\section{Numerical Results}
In this section, we first compare the proposed designs with the exhaustive search (ES) solution obtained by an exhaustively searching algorithm and compare their CPU running time numerically. Then, we compare the performance of the proposed reflecting and precoding designs with existing designs which aims to maximize SNR, that is, SDP-Reflecting and Eigen-Precoding in variously configured $\left(N_r, N, N_t, N_s, M, K\right)$ systems, where $K$ is the Rician fading parameter. All experiments are performed on a PC with a 3.7GHz W-2145 CPU and 32GB RAM. By adopting Alternating Optimization with different reflecting scheme and precoding scheme combinations, vMSER-MMED, SDP-Eigen, SDP-MMED, and vMSER-Eigen are simulated and compared. Then, we demonstrate the SER performance with a large number of reflecting elements to investigate the impact of the size of the surface on the performance. Furthermore, the system performance using perfect CSI and imperfect CSI is compared to investigate the impact of channel estimation errors. Finally, we make a fair comparison between the RIS-assisted system and a full-duplex Amplify-and-Forward (AF)-relaying system.

\subsection{CPU Running Time Comparison and Superiority of the Proposed Schemes}
In Fig. \ref{fig2}, we compare the proposed Reflecting-Precoding combinations with ES solution in $(1, 2, 2, 1, 2, 5)$ system under the same channel realization. The ES solution can be obtained by searching each optimal reflecting elements from the feasible set $C_{\rm{{ESR}}}$ and searching every optimal element of the precoding matrix from the feasible set $C_{\rm{{ESP}}}$. Since the computational complexity of the ES solution exponentially increases with the number of variables $2N_tN_s+N$, thus it only can be realized in systems with small $N_t$, $N_s$ and $N$. To verify the computation efficiency, we present CPU running time comparisons, which is shown in Table I over the same channel for each algorithm. It is seen that all combinations can achieve almost optimal performance compared to the ES solution, which validates the effectiveness of the proposed algorithms. From Table I, it is observed that eMSER-combinations take much longer time than others since the eMSER-Reflecting needs more time to obtain the sub-optimal solution. However, the vMSER-combinations can achieve close performance to eMSER-combination with much lower computation complexity. Thus, vMSER-combinations would be more suitable for practical implementation. Compared with random design, we conclude that appropriate design for the RIS-assisted system can improve system performance significantly. This results from the fact that more signal power can be centered on the direction of the receiver with the well-designed precoding and reflecting. The number of iterations that the Alternating Optimization converged is one of the key factors dominating the computational complexity, which shown in Fig. \ref{fignew}. It is observed  that that Alternating Optimization adopting MMED-Precoding scheme needs more iterations than MSER-Precoding scheme. Moreover, we find that the average number of iterations required for the Alternating Optimization increases as $N$ and $N_t$ increase. This is because  the number of variables and the number of constraints in the optimization problem increase as $N$ and $N_t$ increase.

\begin{figure}[t]
\centering
\includegraphics[width= 3.4in,angle=0]{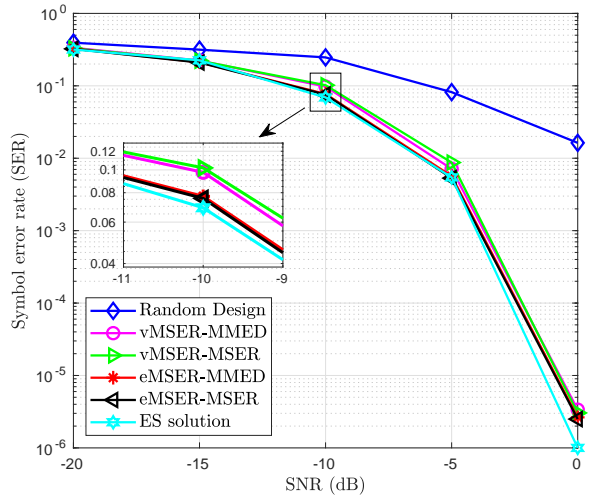}
\caption{SER comparisons among different Reflecting-Precoding combinations.}
\label{fig2}
\end{figure}

\begin{figure}[t]
\centering
\includegraphics[width= 3.4in,angle=0]{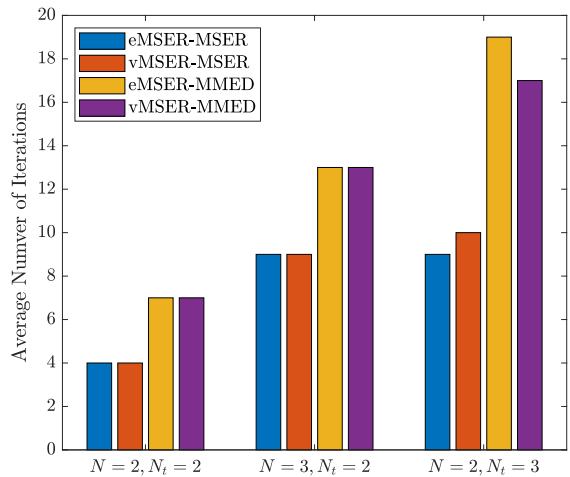}
\caption{Average number of iterations that different Reflecting-Precoding combinations takes to converge.}
\label{fignew}
\end{figure}

\begin{small}
\begin{table*}[!t]
\centering
\begin{tabular}{ P{2.5cm}||P{10cm}|P{2.5cm}  }
 \hline\hline Reflecting Schemes & Complexity Order &CPU Running Time (Seconds)\\ [0.6ex]
 \hline\hline
  eMSER-MSER  &  $\mathcal{O}\left[ M^{2N_s}\left(N^3N_r+NN_tN_s+N^2N_t+NN_rN_t+N_rN_t^2N_s^2\right)\right]$ &6.2408\\[1.5ex]
 \hline
  eMSER-MMED  &  $\mathcal{O}\left[ M^{2N_s}\left(N^3N_r+NN_tN_s+N^2N_t+NN_rN_t+N_rN_t^2N_s^2\right)\right]$ & 6.8444\\[1.5ex]
 \hline
  vMSER-MSER &  $\mathcal{O}\left[ M^{2N_s}\left(N^2N_r+N_tN_s+NN_t+N_rNN_t+ N_rN_t^2N_s^2\right) + N\sum\limits_{n = 1}^{{N_p}} {{p^{\left( n \right)}}} \right]$ &0.8149\\[1.5ex]
 \hline
  vMSER-MMED  &  $\mathcal{O}\left[ M^{2N_s}\left(N^2N_r+N_tN_s+NN_t+N_rNN_t+ N_rN_t^2N_s^2\right) + N\sum\limits_{n = 1}^{{N_p}} {{p^{\left( n \right)}}} \right]$ &0.8123\\[1.5ex]
 \hline
 Exhaustive search &$C_{\rm{ESR}}^{N}C_{\rm{ESP}}^{2N_tN_s}$ & 15,954\\[1.5ex]
 \hline
\end{tabular}
\caption{Computational complexity and CPU running time for different Reflecting-Precoding combinations.}
\end{table*}
\end{small}

\subsection{Comparison with Existing Algorithms}
To make comparisons with other existing algorithms, we introduce them firstly. 
By using the exponential upper bound of $Q$-function, another SER upper bound can be given as 
\begin{align}\label{pa}
{P_S}\left( {\pmb{\Phi} ,\textbf{F}} \right) &\le \frac{1}{{M^{N_s}}}\sum\limits_{i = 1}^{M^{N_s}} {\sum\limits_{j = 1,j \ne i}^{M^{N_s}} {\rm{exp}}\left(-\frac{\rho {d_{ij}^2\left( {\pmb{\Phi} ,\textbf{F}} \right)}}{4\sigma^2}\right)}.
\end{align}
Since the RIS adds new supplementary links to maintain the communication link, the overall system performance can be increased by the in-direct multiple-path without additional interference leading to low required SNR. For this communication situation, the exponential function can be expanded by Taylor expansion $e^{x}=\sum_{n=0}^{\infty} \frac{x^{n}}{n !}$. By only taking the first two terms as ${\rm{exp}}\left(-\frac{\rho {d_{ij}^2\left( {\pmb{\Phi} ,\textbf{F}} \right)}}{4\sigma^2}\right) \approx 1-\frac{\rho {d_{ij}^2\left( {\pmb{\Phi} ,\textbf{F}} \right)}}{4\sigma^2}$, \eqref{pa} becomes ${P_S}\left( {\pmb{\Phi} ,\textbf{F}} \right) \le -\frac{\rho}{4{M^{N_s}}\sigma^2}\sum\limits_{i = 1}^{M^{N_s}} {\sum\limits_{j = 1,j \ne i}^{M^{N_s}} {d_{ij}^2\left( {\pmb{\Phi} ,\textbf{F}} \right)}}+{M^{N_s}}-1$. Based on the given $\textbf{F}$, the objective function of the optimization problem can be changed for maximizing $\sum\limits_{i = 1}^{M^{N_s}} {\sum\limits_{j = 1,j \ne i}^{M^{N_s}} {d_{ij}^2\left( {\pmb{\phi} } \right)}}$. Based on \eqref{dis}, the optimization problem can be written as
\begin{equation}
\begin{split}
\left(\textbf{P5}\right): \;\;\;\rm{Given}: &~\textbf{H}_1, \textbf{H}_2, \textbf{H}_d, \mathcal{S}_M, \textbf{F}\\
\rm{Find}: &~\pmb{\phi}\\
\rm{Maximize}: &~\pmb{\phi}^H \pmb{\Gamma} \pmb{\phi}+\pmb{\gamma}^H\pmb{\phi}+\pmb{\phi}^H\pmb{\gamma} \\
\rm{Subject\;to}: &~\left| {{\phi _i}} \right| = 1, \forall i = 1,...,N,
\end{split}
\end{equation}
where $\pmb{\Gamma} =  \textbf{R}_{\textbf{H}2}\odot \left(\sum\limits_{i = 1}^{M^{N_s}} {\sum\limits_{j = 1,j \ne i}^{M^{N_s}} \Delta {\textbf{Q}_{ij}^T}}\right)$ and $\pmb{\gamma} =\sum\limits_{i = 1}^{M^{N_s}} {\sum\limits_{j = 1,j \ne i}^{M^{N_s}}\textbf{a}_{ij}^C} $. This problem can be treated as a homogeneous QCQP, which is similar to the (P3) shown in \cite{QW} and can be solved in the way shown in \cite{QW}. We named the solution as SDP-Reflecting, whose computational complexity can be analyzed to be $\mathcal{C}_{\rm{SDP}}
= \mathcal{O}\left[N^4+M^{2N_s}N\left(N^2N_tN_r+N_tN_s\right)\right]$ \cite{QW}.

In the following, we compare the system SER of proposed reflecting schemes with SDP-Reflecting in $\left(3, 2, 4, 2, 4, 3\right)$ system. As can be seen from Fig. \ref{fig3}, eMSER-Reflecting and vMSER-Reflecting can achieve much better performance than SDP-Reflecting. SDP-Reflecting is about $2-3$ dB worse than the other two algorithms in the high SNR regime since the exponential approximation applied in \eqref{pa} is close to the original $Q$-function only in the low SNR regime. As a result, the performance gap increases with the increasing SNR in the depicted SNR regime. It is seen that the proposed two reflecting schemes provide apparent favorable SER performance compared to the one with random reflecting. This implies that the phase shift controller can have a significant influence on the system SER. Moreover, it shows that the union bound on SER in \eqref{ob} is quite tight, which implies that the proposed designs minimizing the upper bound of SER will directly minimize the real SER.

\begin{figure}[t]
\centering
\includegraphics[width= 3.4in,angle=0]{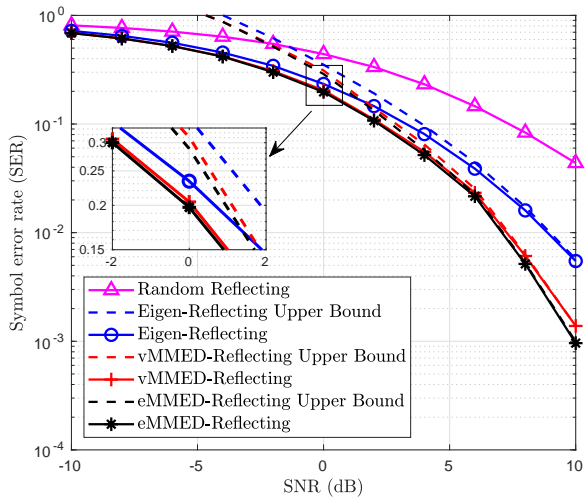}
\caption{SER comparisons among the proposed reflecting schemes and SDP-Reflecting.}
\label{fig3}
\end{figure}

Most RIS-assisted literature optimized precoder based on received signal power, SNR, EE without considering the input vectors. Based on the fact, this problem can be transformed into maximizing the received signal power subject to the maximum transmit power constraint at the transmitter, which can be formulated as
\begin{equation}
\begin{split}
\left(\textbf{P6}\right): \;\;\;\rm{Given}: &~\textbf{H}, \mathcal{S}_M, {\pmb{\Phi}} \\
\rm{Find}: &~\textbf{F}\\
\rm{Maximize}: & ~\left\| \textbf{HF} \right\|_F^2\\
\rm{Subject\;to}: &~\text{tr}\left( {\textbf{Q}} \right) \le {P_{\max }}.
\end{split}
\end{equation}

Since $\textbf{F}$ is full-rank matrix, we can easily obtain the solution as $\textbf{F} = \sqrt{\frac{P_{\max}}{N_s}} \textbf{W}$, where the $k$-th column of $\textbf{W}$ is the $k$-th eigenvector of $\textbf{H}^H\textbf{H}$ corresponding to $k$-th eigenvalue with $\lambda_k \ge \lambda_{k+1}$. We named this solution as Eigen-Precoding in the following comparison. The computational complexity of the Eigen-Precoding scheme is consumed by computing $\textbf{H}^H\textbf{H}$ and its eigenvectors, which is $\mathcal{C}_{\rm{EP}}=\mathcal{O}\left[N^2N_r+N_rNN_t+N_t^2N_r+N_t^3\right]$,
where $N^2N_r+N_rNN_t$ is the computation complexity of $\textbf{H}$, $N_t^2N_r$ is the computation complexity of $\textbf{H}^H\textbf{H}$ and $N_t^3$ is the eigenvalue decomposition complexity. In the following, we compare the system SER of proposed precoding schemes with Eigen-Precoding in $\left(3, 20, 3, 2, 4, 1\right)$ system.

\begin{figure}[t]
\centering
\includegraphics[width= 3.4in,angle=0]{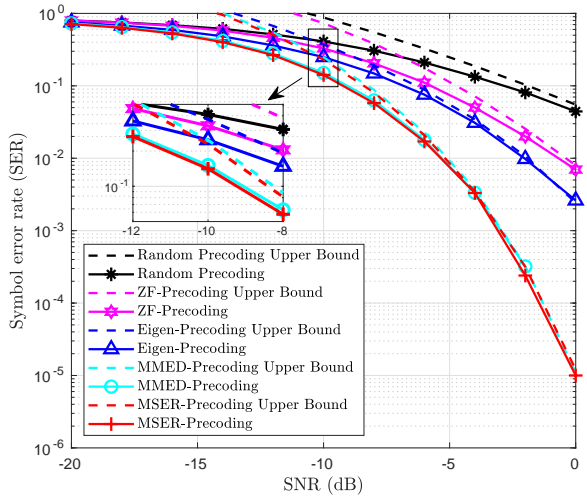}
\caption{SER comparisons among the proposed precoding schemes and Eigen-Precoding.}
\label{fig4}
\end{figure}

As can be seen from Fig. \ref{fig4}, all the proposed schemes can provide lower SER than random precoding. We can see that the proposed MSER-Precoding scheme has a slightly better performance than MMED-Precoding scheme. As analyzed in Section V, the MMED-Precoding is of lower complexity than MSER-Precoding. Thus, in realistic systems, MMED-Precoding is a more appealing scheme compared to MSER-Precoding. Eigen-Precoding is much faster than two other schemes but much worse performance. It is obvious that a big performance gap occurs between the Eigen-Precoding and the other two schemes under the given system setup. Although Eigen-Precoding has significant lower computation complexity, it is about 4-5 dB worse than MSER-Precoding and MMER-Precoding. It is validated that the algorithm designed with the Gaussian input assumption should be reconsidered when applied to systems with finite alphabet input. By comparing the proposed schemes with the classical zero-forcing (ZF) precoding, we observe that the proposed schemes can bring significant performance gain in SER. The reason behind the phenomenon is that our proposed precoding schemes aim to minimize the SER directly, while ZF precoding is designed to null the interference. Compared to the performance of random precoder, we can see the importance of well-designed precoding schemes. 

\begin{figure}[t]
\centering
\includegraphics[width= 3.4in,angle=0]{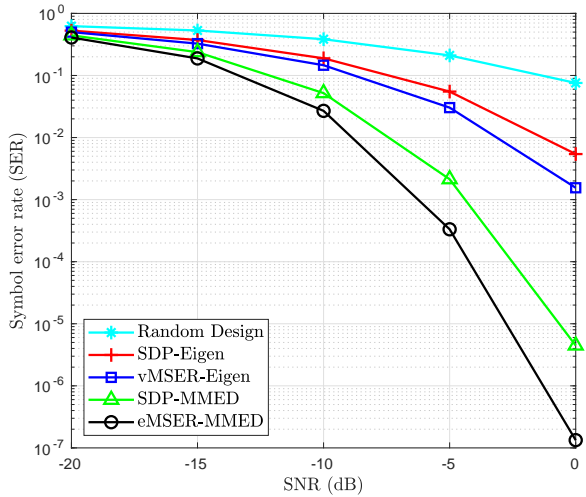}
\caption{SER comparisons among vMSER-MMED, vMSER-Eigen, SDP-MMED and SDP-Eigen.}
\label{fig5}
\end{figure}

By processing the SDP-Reflecting and Eigen-Precoding alternatively in $\left(2, 5, 3, 2, 4, 2\right)$ system, we show the simulation results in Fig. \ref{fig5}. Based on the results of Figs. 3-4, there is no doubt that SDP-Eigen would achieve worse performance. Since the SDP-Reflecting maximizing the approximation in high SNR regime and Eigen-Precoding maximizing the SNR is not equivalent to maximizing the Euclidean distances which directly affects the SER. That is why the performance of SDP-Eigen is not comparable to proposed schemes, which aim to minimize SER directly. Even with Alternating Optimization, the system SERs with Eigen-Precoding is not favorable since it does not take inputs into consideration. It proves that the algorithm designed for continues or only finite inputs cannot replace the optimization that needs to be done for limited discrete data. The influence of input symbols taken from finite and discrete constellation should be taken into consideration. 

\subsection{Impact of the Number of Reflecting Elements}
\begin{figure}[t]
\centering
\includegraphics[width= 3.4in,angle=0]{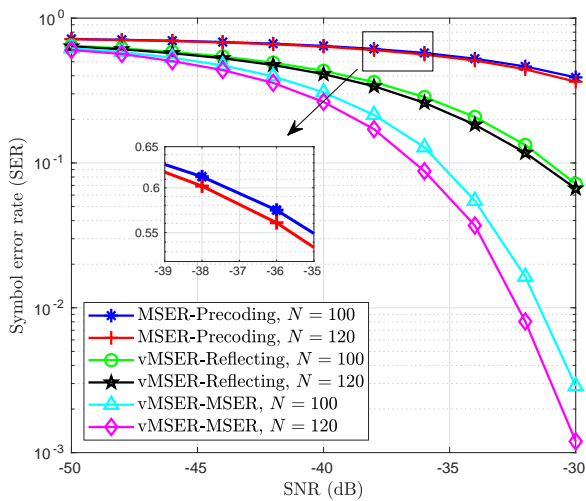}
\caption{SER comparisons among vMSER-Reflecting and MSER-Precoding, vMSER-MSER with various large $N$.}
\label{fig6}
\end{figure}
In Fig. \ref{fig6}, we show the influence of $N$ on the phase shift design, precoding design and joint design in $\left(3, 100, 3, 2, 2, 2\right)$ and $\left(3, 120, 3, 2, 2, 2\right)$ system. From Fig. \ref{fig6}, we can see that phase shift design will play a more important role than precoding design on affecting SER when $N$ is large. There is no obvious SER gain with a larger $N$ if the reflecting is not carefully designed. Obviously, in Fig. \ref{fig6}, the joint design can gain optimal performance at the cost of more iteration to gain a better solution. There is 1 dB SER gain with 20 more reflecting elements if both the reflecting and precoding are optimized. In conclusion, the system SER can be significantly decreased by equipping more reflecting elements at the RIS and using the proposed reflecting design. $N$ will not provide any SER gain without phase shift design. 

\subsection{Superiority in Presence of CSI Estimation Errors}
Although we were assuming all involved CSI is perfectly known, the channel estimation for the RIS-assisted system is challenging because of its massive number of passive elements without any signal processing capability. The authors in \cite{BZ} proposed a practical transmission protocol based on the received pilot signals from the user to execute channel estimation. The authors in \cite{ZQ} presented a two-stage algorithm that includes the sparse matrix factorization stage and the matrix completion stage to do channel estimation, which was shown to be quite accurate. However, it is still an ideal assumption to have perfect and timely channel information in practical application. Based on the fact, the performance of proposed schemes in the presence of channel estimation errors should be probed into.
\begin{figure}[t]
\centering
\includegraphics[width= 3.4in,angle=0]{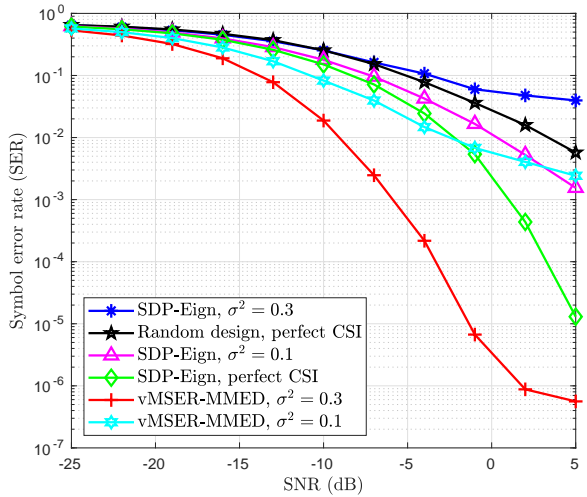}
\caption{SER comparisons in presence of CSI estimation errors.}
\label{fig7}
\end{figure}
To show the superiority of the proposed schemes in the presence of CSI estimation errors, we further evaluate the error performance of our schemes with imperfect CSI. SDP-Eigen scheme and vMSER-MMED scheme are chosen as examples. Such imperfection originates from channel estimation and/or feedback errors. The model of imperfect CSI is given by $\textbf{H'} = \textbf{H}+\textbf{H}_e$, where $\textbf{H}_e$ is the error matrix whose entries follow and i.i.d circularly symmetric complex Gaussian distribution $\mathcal{CN}\left( {\textbf{0},\sigma_e^2 }\right)$. The error performance of the joint SDP-Eigen scheme with different error levels $\sigma_e^2 = 0$ (i.e., perfect CSI), $\sigma_e^2= 0.1$, $\sigma_e^2 = 0.3$ in $\left(4, 4, 3, 2, 2, 1\right)$ system is compared in the Fig. \ref{fig7}. As expected, the error performance degrades with imperfect CSI. The SDP-Eigen scheme and vMSER-MMED scheme with imperfect CSI will achieve a floor as SNR grows. It can also be seen that proposed schemes with imperfect CSI still achieve an effective system SER optimization compared to a random joint design with perfect CSI. For example, the SDP-Eigen scheme achieves about 2-3 dB SER gain for $\sigma_e^2 = 0.1$ while vMSER-MMED scheme achieves about 10-15 dB SER gain for $\sigma_e^2 = 0.1$ and about 5-6 dB SER gain for $\sigma_e^2 = 0.3$ compared to random design with perfect CSI, which demonstrates the superiority of our design. Furthermore, vMSER-MMED scheme achieves about 2-3 dB SER gain for $\sigma_e^2 = 0.3$ in low SNR regime and about 6-7 dB SER gain for $\sigma_e^2 = 0.1$ compared with SDP-Eigen scheme with perfect CSI. This means vMSER-MMED scheme can achieve greatly better performance even with imperfect CSI since the Eigen-Precoding could not achieve a favorable performance with multiple discrete input data streams.

\subsection{Comparison with AF Relay Cooperating System}
In this subsection, we compare a   $\left(3, 5, 3, 2, 4, 1\right)$ RIS-assisted system with a full-duplex AF relaying system . The reason why using RIS is because of its low complexity and low cost compared to the relay-based system \cite{KN}. Relay stations equipped with active electronic components to receive and forward with a dedicated power source. The deployment of relays is costly and power-consuming, especially for realizing multiple-antenna designs. 
\begin{figure}[t]
\centering
\includegraphics[width= 3.4in,angle=0]{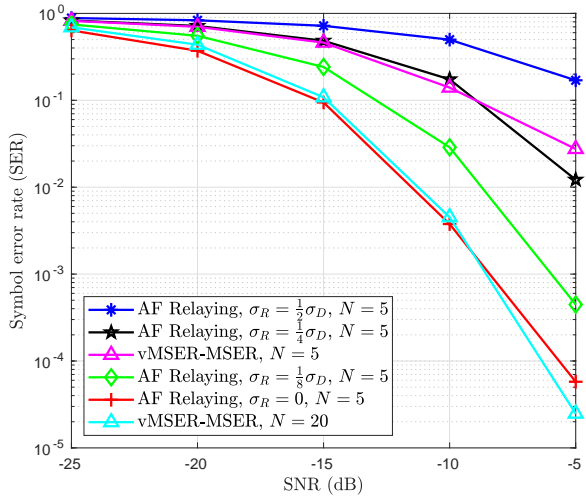}
\caption{SER comparisons for RIS-assited system and AF relaying system.}
\label{fig8}
\end{figure}

For fair comparison in performance,  we assume there is a conventional $N$-antenna AF relay in the place of the RIS structure under same channel realizations and input data streams, and the received signal at the destination can be expressed as $\textbf{y}_D =  \sqrt{\rho} \left(\textbf{H}_2 \textbf{V}\textbf{H}_1 +\textbf{H}_d\right)\textbf{F} \textbf{s} + \textbf{H}_2 \textbf{V}\textbf{n}_R +\textbf{n}_D$, where $ \textbf{n}_R \sim \mathcal{CN}\left( {\textbf{0},\sigma_R^2{\textbf{I}_{{N_r}}}} \right)$, $ \textbf{n}_D \sim \mathcal{CN}\left( {\textbf{0},\sigma_D^2{\textbf{I}_{{N_r}}}} \right),$ $\textbf{V}$ is the diagonal AF matrix and $\textbf{n}_R$ is the additive noise at the AF relay. We jointly optimize $\textbf{V}$ and $\textbf{F}$ alternatively. The optimization for precoder and AF matrix is similar to the optimization for precoder of RIS-assisted system with power constraint. As can be seen from Fig. \ref{fig8}, AF relaying system can achieve much lower SER compared to the same $N$ RIS-assisted system when there is no noise at $R$. The reason is that AF relay consumes additional power to forward received signals, while there is no power gain in RISs. However, we can also see that the performance gap between the RIS-assisted system and AF relaying systems decreases as the variance of noise at relay increasing. As AF-relaying introduces more noise power, the superiority of performance would eliminate. Moreover, the RIS-assisted system outperforms the AF-relaying system by equipping more reflecting elements. Another $10$ reflecting elements can offer more $7-8$ dB SER gain and achieve almost the same performance to the noise-free AF relaying system. It means the system performance can be greatly improved in a low cost and energy-saving way.

\section{Conclusion}
In the paper, several joint reflecting and precoding schemes were proposed to minimize the SER in point-to-point RIS-assisted MIMO systems assuming finite alphabet input. The proposed reflecting and precoding schemes can achieve a favorable system performance with different computation complexity. A much lower complex algorithm based on SNR without considering the influence of input data streams can not be used to replace the finite discrete input system. Simulation results demonstrate that our proposed algorithm can significantly improve SER performance with different scheme combinations while reflecting schemes will play a more critical role than precoding schemes. The RIS-assisted system with proper reflecting and precoding design has good robustness with imperfect CSI. By equipping more reflecting elements, it can improve the system performance significantly and surpass the one of relaying systems with consuming additional power.


\begin{thebibliography}{1}
\bibitem{SZ}
S. Zhang, Q. Wu, S. Xu, and G. Y. Li, ``Fundamental green tradeoffs: Progresses, challenges, and impacts on 5G networks," \emph{IEEE Commun.Surveys Tuts.}, vol. 19, no. 1, pp. 33-56, First Quarter 2017.

\bibitem{MDR}
M. Di Renzo, M. Debbah, D.-T. Phan-Huy, A. Zappone, M.-S.Alouini, C. Yuen, V. Sciancalepore, G. C. Alexandropoulos, J. Hoydis, and H. Gacanin, ``Smart radio environments empowered by AI reconfigurable meta-surfaces: An idea whose time has come," \emph{EURASIP J. Wireless Commun. Netw.}, vol. 2019, p. 129, May 2019.

\bibitem{XT}
X. Tan, Z. Sun, D. Koutsonikolas, and J. M. Jornet, ``Enabling indoor mobile millimeter-wave networks based on smart reflect-arrays," in \emph{Proc. IEEE INFOCOM}, Honolulu, USA, Apr. 2018, pp. 270-278.

\bibitem{SV}
S. V. Hum and J. Perruisseau-Carrier, ``Reconfigurable reflectarrays and array lenses for dynamic antenna beam control: A review," \emph{IEEE Trans. Antennas Propagat.}, vol. 62, no. 1, pp. 183-198, Jan. 2014.

\bibitem{SF}
S. Foo, ``Liquid-crystal reconfigurable metasurface reflectors," in \emph{Proc. IEEE ISAP}, California, USA, July 2017, pp. 2069-2070.

\bibitem{CL}
C. Liaskos, A. Tsioliaridou, A. Pitsillides, S. Ioannidis, and I. F. Akyildiz, ``Using any surface to realize a new paradigm for wireless communications," \emph{Commun. ACM}, vol. 61, pp. 30-33, Oct. 2018.

\bibitem{TJ}
T. Jun Cui, M. Q. Qi, X. Wan, J. Zhao, and Q. Cheng, ``Coding metamaterials, digital metamaterials and programmable metamaterials," \emph{Light, Sci. Appl.}, vol. 3, no. 10, p. e218, Oct. 2014.

\bibitem{BS}
B. Sainath and N. B. Mehta, ``Generalizing the amplify-and-forward relay gain model: An optimal SEP perspective," \emph{IEEE Trans. Wireless Commun.}, vol. 11, no. 11, pp. 4118-4127, Nov. 2012.

\bibitem{GY}
G. Yang, C. K. Ho and Y. L. Guan, ``Multi-antenna wireless energy transfer for backscatter communication systems," \emph{IEEE J. Sel. Areas Commun.}, vol. 33, no. 12, pp. 2974-2987, Dec. 2015.

\bibitem{SH}
S. Hu, F. Rusek, and O. Edfors, ``Beyond massive MIMO: The potential of data transmission with large intelligent surfaces," \emph{IEEE Trans. Signal Process.}, vol. 66, no. 10, pp. 2746-2758, May 2018.

\bibitem{LS}
L. Subrt and P. Pechac, ``Intelligent walls as autonomous parts of smart indoor environments," \emph{IET Commun.}, vol. 6,no. 8, pp. 1004-1010, May 2012.

\bibitem{NTT}
``NTT DoCoMo and Metawave announce successful demonstration of 28GHz-band 5G using world’s first meta-structure technology," [Online]. Available: https://www.marketwatch.com/press-release/ntt-docomo-and-metawave-announce-successful-demonstration-of-28ghz-band-5g-using-worlds-first-meta-structure-technology-2018-12-04, July 7, 2019.

\bibitem{BZ}
B. Zheng, R. Zhang, ``Intelligent reflecting surface-enhanced OFDM: Channel estimation and reflection optimization," to appear in \emph{IEEE Wireless Commun. Lett.}, 2020.

\bibitem{ZQ}
Z. He and X. Yuan, "Cascaded Channel Estimation for Large Intelligent Metasurface Assisted Massive MIMO," \emph{IEEE Wireless Commun. Lett.}, vol. 9, no. 2, pp. 210-214, Feb. 2020.

\bibitem{CP}
C. Pan, H. Ren, K. Wang, W. Xu, M. Elkashlan, A. Nallanathan and L. Hanzo, ``Intelligent reflecting surface for multicell MIMO communications," \emph{CoRR}, vol. abs/1907.10864. 2019. [Online]. Available: http://.org/abs/1907.10864.

\bibitem{EBR}
E. Bjornson, O. Ozdogan, and E. G. Larsson, ``Intelligent reflecting surface versus decode-and-forward: How large surfaces are needed to beat relaying?” \emph{IEEE Wireless Commun. Lett.}, vol. 9, no. 2, pp. 244-248, Feb. 2020.

\bibitem{EBM}
E. Bjornson and L. Sanguinetti, ``Demystifying the power scaling law of intelligent reflecting surfaces and metasurfaces," in \emph{IEEE CAMSAP}, Dec. 2019.

\bibitem{QW3}
Q. Wu and R. Zhang, ``Weighted sum power maximization for intelligent reflecting surface aided SWIPT," to appear in \emph{IEEE Wireless Commun. Lett.}, 2020.

\bibitem{CP1}
C. Pan, H. Ren, K. Wang, M. Elkashlan, A. Nallanathan, J. Wang and L. Hanzo, `` Intelligent reflecting surface enhanced MIMO broadcasting for simultaneous wireless information and power transfer," in \emph{CoRR}. vol. abs/1908.04863. 2019. [Online]. Available: http://arxiv.org/abs/1908.04863.

\bibitem{PW}
P. Wang, J. Fang, X. Yuan, Z. Chen, H. Duan and H. Li, ``Intelligent reflecting surface-assisted millimeter wave communications: Joint active and passive precoding design," \emph{CoRR}, vol. abs/1908.10734. 2019. [Online]. Available: http://.org/abs/1908.10734.

\bibitem{WC}
W. Chen, X. Ma, Z. Li and N. Kuang, ``Sum-rate maximization for intelligent reflecting surface based Terahertz communication systems," in \emph{Proc. ICCC Workshops}, Changchun, China, 2019, pp. 153-157.

\bibitem{QW}
Q. Wu and R. Zhang, ``Intelligent reflecting surface enhanced wireless network: Joint active and passive beamforming design," in \emph{Proc. IEEE GLOBECOM}, Abu Dhabi, UAE, Dec. 2018, pp. 1-6. 

\bibitem{QW2}
Q. Wu and R. Zhang, ``Beamforming optimization for intelligent reflecting surface with discrete phase shifts," in \emph{Proc. IEEE ICASSP}, Brighton, UK, May 2019, pp. 7830-7833.


\bibitem{QI}
Q.-U.-A. Nadeem, A. Kammoun, A. Chaaban, M. Debbah, and M.-S. Alouini ``Asymptotic analysis of large intelligent surface assisted MIMO communication," \emph{CoRR}, vol. abs/1903.08127, 2019. [Online]. Available: http://arxiv.org/abs/1903.08127.

\bibitem{CH}
C. Huang, A. Zappone, G. C. Alexandropoulos, M. Debbah and C. Yuen, ``Reconfigurable intelligent surfaces for energy efficiency in wireless communication," \emph{IEEE Trans. Wireless Commun.}, vol. 18, no. 8, pp. 4157-4170, Aug. 2019.

\bibitem{CH1}
C. Huang, G. C. Alexandropoulos, A. Zappone, M. Debbah, and C. Yuen, ``Energy efficient multi-user MISO communication using low resolution large intelligent surfaces," in \emph{Proc. IEEE GLOBECOM}, Abu Dhabi, UAE, Dec. 2018, pp. 1-6. 

\bibitem{YH}
Y. Han, W. Tang, S. Jin and C. Wen and X. Ma, ``Large intelligent surface-assisted wireless communication exploiting statistical CSI," \emph{IEEE Trans. Veh. Technol.}, vol. 68, no. 8, pp. 8238-8242, Aug. 2019.

\bibitem{MJ}
M. Jung, S. Walid, R. J. Young, K. Gyuyeol and C. Sooyong, ``Performance analysis of large intelligence surfaces (LISs): Asymptotic data rate and channel hardening effects," \emph{IEEE Trans. Wireless Commun.}, vol. 19, no. 3, pp. 2052-2065, March 2020. 

\bibitem{MJ1}
M. Jung, S. Walid and K. Gyuyeol, ``Performance analysis of large intelligent surfaces (LISs): Uplink spectral efficiency and pilot training," \emph{CoRR}, vol. abs/1904.00453, 2019. [Online]. Available: http://arxiv.org/abs/1904.00453. 

\bibitem{MJ3}
M. Jung, W. Saad, Y. Jang, G. Kong, and S. Choi, ``Uplink data rate in large intelligent surfaces: Asymptotic analysis under channel estimation errors," in \emph{Proc. IEEE SPAWC} ,Cannes, France, July 2019, pp. 1-6. 

\bibitem{EB}
E. Basar, ``Transmission through large intelligent surfaces: A new frontier in wireless communications," in \emph{Proc. IEEE EuCNC},  Valencia, Spain, 2019, pp. 112-117.

\bibitem{XT1}
X. Tan, Z. Sun, J. M. Jornet, and D. Pados, ``Increasing indoor spectrum sharing capacity using smart reflect-array," in \emph{Proc. IEEE ICC}, Kuala Lumpur, Malaysia, May 2016, pp. 1-6.

\bibitem{HG}
H. Guo, Y. C. Liang, J. Chen, and E. G. Larsson, ``Weighted sum-rate optimization for intelligent reflecting surface enhanced wireless networks," \emph{CoRR}, vol. abs/1905.07920. 2019. [Online]. Available: http://.org/abs/1905.07920.

\bibitem{SA}
S. Abeywickrama, R. Zhang and C. Yuen, ``Intelligent reflecting surface: Practical phase shift model and beamforming optimization," \emph{CoRR}, vol. abs/1907.06002. 2019. [Online]. Available: http://.org/abs/1907.06002.

\bibitem{EB1}
E. Basar, ``Large intelligent surface-based index modulation: A new beyond MIMO paradigm for 6G," \emph{CoRR}, vol. abs/1904.06704, 2019. [Online]. Available: http://arxiv.org/abs/1904.06704.

\bibitem{CH2}
C. Huang, A. Zappone, M. Debbah, and C. Yuen, ``Achievable rate maximization by passive intelligent mirrors," in \emph{Proc. IEEE ICASSP}, Calgary, Canada, Apr. 2018, pp. 1-6.

\bibitem{GX}
X. Guan , Q. Wu, R. Zhang, ``Intelligent reflecting surface assisted secrecy communication via joint beamforming and jamming," \emph{CoRR}, vol. abs/1907.12839. 2019. [Online]. Available: http://.org/abs/1907.12839.

\bibitem{HS}
H. Shen, W. Xu, S. Gong, Z. He and C. Zhao, ``Secrecy rate maximization for intelligent reflecting surface assisted multi-antenna communications," \emph{IEEE Commun. Lett.}, vol. 23, no. 9, pp. 1488-1492, Sept. 2019.

\bibitem{MC}
M. Cui, G. Zhang and R. Zhang, ``Secure wireless communication via intelligent reflecting surface," \emph{IEEE Wireless Commun. Lett.}, vol. 8, no. 5, pp. 1410-1414, Oct. 2019.

\bibitem{WY}
Y. Wu, C. Xiao, Z. Ding, X. Gao and S. Jin, ``A survey on MIMO transmission with finite input signals: Technical challenges, advances, and future trends," in \emph{Proc. IEEE}, vol. 106, no. 10, pp. 1779-1833, Oct. 2018.

\bibitem{QWW}
Q. Wu and R. Zhang, ``Intelligent Reflecting Surface Enhanced Wireless Network via Joint Active and Passive Beamforming," \emph{IEEE Trans. Wireless Commun.}, vol. 18, no. 11, pp. 5394-5409, Nov. 2019.

\bibitem{XZ}
X. Zhang, \emph{Matrix Analysis and Applications}. Tsinghua, China: Tsinghua Univ. Press, 2004.


\bibitem{LH}
L. He, J. Wang and J. Song, ``Spatial modulation for more spatial multiplexing: RF-chain-limited generalized spatial modulation aided MM-Wave MIMO With hybrid precoding," \emph{IEEE Trans. Commun.}, vol. 66, no. 3, pp. 986-998, Mar. 2018.

\bibitem{SP}
S. P. Boyd and L. Vandenberghe, \emph{Convex Optimization}. Cambridge, U.K.: Cambridge Univ. Press, 2004.

\bibitem{WW}
W. Wang and W. Zhang, ``Diagonal precoder designs for spatial modulation," in \emph{Proc. IEEE ICC}, London, UK, June 2015, pp. 2411-2415.

\bibitem{PC}
P. Cheng, Z. Chen, J. A. Zhang, Y. Li and B. Vucetic, ``A unified precoding scheme for generalized spatial modulation," \emph{IEEE Trans. Commun.}, vol. 66, no. 6, pp. 2502-2514, June 2018.

\bibitem{AA}
A. Antoniou and W.-S. Lu, \emph{Practical Optimization: Algorithms and Engineering Applications}. New York, NY, USA: Springer, 2007.

\bibitem{MG}
M. Grant and S. Boyd, ``CVX: MATLAB software for disciplined convex programming," 2016. [Online] Available: http://cvxr.com/cvx.

\bibitem{shuai}
S. Guo, H. Zhang, P. Zhang, S. Dang, C. Liang and M. S. Alouini, ``Signal shaping for generalized spatial modulation and generalized quadrature spatial modulation," to appear in \emph{IEEE Trans. Wireless Commun.}, vol. 18, no. 8, pp. 4047-4059, Aug. 2019.

\bibitem{AS}
A. Skajaa, ``Limited Memory BFGS for Non-smooth Optimization," M.S. Thesis, Dept. Comput. Sci. Math., Courant Inst. Math. Sci., New York, NY, USA, 2010.

\bibitem{SB}
S. Boyd, N. Parikh, E. Chu, B. Peleato, and J. Eckstein, ``Distributed optimization and statistical learning via the alternating direction method of multipliers," \emph{Found. Trends Mach. Learn.}, vol. 3, no. 1, pp. 1-122, Jan. 2011.

\bibitem{DT}
D. Tse and P. Viswanath, \emph{Fundamentals of Wireless Communication}. Cambridge, U.K.: Cambridge Univ. Press, 2005.

\bibitem{KN}
K. Ntontin, M. Di Renzo, J. Song, F. Lazarakis, J. de Rosny, D.-T.Phan-Huy, O. Simeone, R. Zhang, M. Debbah, G. Leroseyet, M. Fink, S. Tretyakov, S. Shamai, ``Re-configurable intelligent surfaces vs. relaying: Differences, similarities, and performance comparison," \emph{CoRR}, vol. abs/1908.08747, 2019. [Online]. Available: http://.org/abs/1908.08747.

\end{thebibliography}
\end{document}